\documentclass[12pt,reqno]{amsart}%[12pt]{extarticle}
\usepackage{amssymb,amsthm,amsmath,amsxtra,dsfont,appendix,bbm,stix}
\usepackage{hyperref} 
\usepackage{mathrsfs}

\numberwithin{equation}{section}
\makeatletter

\makeatletter
\def\@tocline#1#2#3#4#5#6#7{\relax
  \ifnum #1>\c@tocdepth % then omit
  \else
    \par \addpenalty\@secpenalty\addvspace{#2}%
    \begingroup \hyphenpenalty\@M
    \@ifempty{#4}{%
      \@tempdima\csname r@tocindent\number#1\endcsname\relax
    }{%
      \@tempdima#4\relax
    }%
    \parindent\z@ \leftskip#3\relax \advance\leftskip\@tempdima\relax
    \rightskip\@pnumwidth plus4em \parfillskip-\@pnumwidth
    #5\leavevmode\hskip-\@tempdima
      \ifcase #1
       \or\or \hskip 1em \or \hskip 2em \else \hskip 3em \fi%
      #6\nobreak\relax
      \dotfill
      \hbox to\@pnumwidth{\@tocpagenum{#7}}
    \par
    \nobreak
    \endgroup
  \fi}
\makeatother

\renewcommand{\i}{\mathrm{i}}

\newcommand{\C}{\mathds{C}}

\newcommand{\1}{\mathds{1}}

\newcommand{\N}{\mathbb{N}}
\newcommand{\R}{\mathds{R}}

\renewcommand{\S}{\mathfrak{S}}

\newcommand{\sgn}{\operatorname{sgn}}

\newtheorem{theorem}{Theorem}[section]

\newtheorem{definition}[theorem]{Definition}

\newtheorem{assumption}[theorem]{Assumption}
\newtheorem{proposition}[theorem]{Proposition}

\newtheorem{lemma}[theorem]{Lemma} 
\newtheorem{remark}[theorem]{Remark}

\newcommand{\bx}{\mathbf{x}}
\newcommand{\bX}{\mathbf{X}}
\newcommand{\ba}{\mathbf{a}}
\newcommand{\by}{\mathbf{y}}
\newcommand{\bz}{\mathbf{z}}
\newcommand{\bw}{\mathbf{w}}
\newcommand{\bW}{\mathbf{W}}

\newcommand{\cH}{\mathcal{H}}

\newcommand{\cK}{\mathcal{K}}

\newcommand{\cZ}{\mathcal{Z}}

\newcommand{\cF}{\mathcal{F}}

\newcommand{\mueq}{\mu_{\rm eq}}
\newcommand{\cEel}{\mathcal{E}^{\rm MF}}
\newcommand{\Eel}{E^{\rm MF}}
\newcommand{\ZN}{\mathcal{Z}_{N}}
\newcommand{\ZNV}{\mathcal{Z}_{N}^V}
\newcommand{\FC}{F^{\rm Corr}}

\newcommand{\ZGin}{\cZ^{\mathrm{Gin}}}

\newcommand{\Int}{\mathcal{I}}
\newcommand{\Kt}{\widetilde{K}}

\numberwithin{equation}{section}

\begin{document}

\title{Free-energy variations for determinantal 2D plasmas with holes}

\author[N. Rougerie]{Nicolas Rougerie}
\address{Ecole Normale Sup\'erieure de Lyon \& CNRS,  UMPA (UMR 5669)}
\email{nicolas.rougerie@ens-lyon.fr}

\date{July, 2026}

\begin{abstract}
We study the Gibbs equilibrium of a classical 2D Coulomb gas in the determinantal case $\beta =2$. The external potential is the sum of a quadratic term and the potential generated by individual charges pinned in several extended groups. This leads to an equilibrium measure (droplet) with flat density and macroscopic holes. We consider ``correlation energy'' (free energy minus its mean-field approximation) expansions, for large particle number $N$. Under the assumptions that the holes are sufficiently small, separated, and far from the droplet's outer boundary, we prove that (i) the correlation energy up to order $1$ is independent of the holes' locations and orientations, and (ii) the difference between the  correlation energies of systems differing by their number of holes essentially consists of ``topological'' $O(\log N)$ and $O(1)$ terms.   
\end{abstract}

\maketitle

\begin{center}
 \emph{Dedicated to Jakob Yngvason, on the occasion of his 80th birthday.}
\end{center}

\tableofcontents

\section{Introduction}

The 2D classical Coulomb gas\footnote{Always understood as the one-component plasma,  hereafter.}, on top of being an emblematic statistical physics model in its own right, is widely studied for its many connections with different fields of physics and mathematics~\cite{Forrester-10,ForByu-25,Klevtsov-16,Lewin-22,Serfaty-24,Rougerie-Elliott}. Of chief interest is the model's behavior for large particle numbers $N$, in particular effects beyond mean-field (MF) theory. Indeed, in the setting of our interest below, the leading order behavior is dictated by a non-linear effective one-particle theory, setting the macroscopic distribution of charges (the droplet). After zooming in on the microscopic inter-particles scale, a thermodynamic limit emerges as a local density approximation (LDA) of the original problem, where the ``local density'' is given by mean-field theory. Fluctuations beyond that are governed by a gaussian free field (GFF) emerging from the LDA. Recent years have seen this picture confirmed in great generality, we refer to~\cite{Serfaty-24}, in particular Section 9 therein for extensive review and references to the literature. Closest to our setting below, see in particular~\cite{AmeHedMak-11,BauBouNikYau-15,BauBouNikYau-16,LebSer-15,LebSer-16}.      

The behavior beyond LDA remains elusive, contrarily to the corresponding question for related 1D models (1D log-gases~\cite{BorGui-13,BorGui-24}). Predictions from the physics literature~\cite{JanManPis-94,JanTri-00,ZabWie-06} pointing to further signatures of the emergent GFF and topological effects have so far been mathematically vindicated only in special determinantal cases (and thus, for a specific temperature choice): on the sphere with holes at  the poles~\cite{ByuForLah-25,ByuCharMorSim-25,CriKui-22,BraDraSafWom-18}, on  Riemann surfaces without boundaries~\cite{KleMaMarWie-17,SheYu-25,Bourgoin-25}, in a radial context~\cite{AmeChaCro-23,ByuKanSeo-23,AllForLahShe-25,AllLah-25}, for a model with at most one hole in the droplet~\cite{ByuSeoYan-24,ByuYanYoo-26}, for special models leading to disconnected droplets~\cite{ByuLeeYan-21,Byun-25} etc ...

Our purpose is to investigate some of the signatures of the conjectured free-energy expansions~\cite{JanManPis-94,JanTri-00,ZabWie-06} in a special (determinantal) model where the droplet is non-radial and can have several holes. We cannot provide a full free-energy expansion, but we obtain clear signatures of the ``topological'' $\log N$ terms of the expansion~\footnote{In the convention we follow, the leading MF term is of order $N^2$, the LDA term of order $N$ being often considered the leading one when dealing with a neutral homogeneous system~\cite{JanManPis-94}.}, and some of the expected invariance features of the $O(1)$ terms. 

Consider $N$ particles in the plane of coordinates $\bX_N = (\bx_1, \ldots,\bx_N) \in \R^{2N}$ with energy
\begin{equation}\label{eq:Hamil}
H_N (\bx_1,\ldots,\bx_N) := \frac{1}{2}\sum_{j=1}^N N V(\bx_j) - \sum_{j<k} \log |\bx_j - \bx_k| 
\end{equation}
and consider the Gibbs state in the determinantal case (inverse temperature $\beta = 2$)
\begin{align}\label{eq:Gibbs}
\rho_{N,V} (\bx_1,\ldots,\bx_N) &= \frac{1}{\ZNV} \exp\left( -2 H_N (\bx_1,\ldots,\bx_N)\right)\nonumber \\
&= \frac{1}{\ZNV} \prod_{1\leq j < k \leq N} |\bx_j - \bx_k| ^2 e^{-N\sum_{j=1}^N V(\bx_j)}.
\end{align}
The logarithmic pairwise interaction corresponds to 2D Coulomb forces, and $V:\R^2 \mapsto \R$ is an external trapping potential, e.g. generated by a fixed charge distribution interacting with the $\bx_j$. By definition $\rho_{N,V}$ minimizes the free-energy functional 
\begin{equation}\label{eq:intro free func}
\cF_{N,V} [\mu] := \int_{\R^{2N}} H_N (\bX_N) \mu (\bX_N) d\bX_N + \frac{1}{2}  \int_{\R^{2N}} \mu (\bX_N) \log \mu (\bX_N) d\bX_N
\end{equation}
amongst probability measures $\mu$ on $\R^{2N}$ (in practice, amongst positive $L^1$-normalized functions). The corresponding infimum is 
\begin{equation}\label{eq:intro free ener}
F_{N,V} = - \frac{1}{2} \log \ZNV 
\end{equation}
and we are interested in large $N$ expansions thereof. Define, for a probability measure $\sigma$ on $\R^2$,  the mean-field energy functional
\begin{equation}\label{eq:intro MFE}
\cEel [\sigma] := \frac{N^2}{2}\int_{\R^2} V(\bx) \sigma (\bx) d\bx - \frac{N^2}{2}\iint_{\R^2 \times \R^2}\sigma(\bx) \log |\bx-\by| \sigma(\by) d\bx d\by
\end{equation}
obtained by inserting an uncorrelated ansatz $\mu = \sigma ^{\otimes N}$ in~\eqref{eq:intro free func} and neglecting the entropy term. Under very mild assumptions, the above has a minimum, denoted $\Eel$, and a minimizer $\mueq$, called the equilibrium measure. In great generality (see the aforementioned references) we have 
$$ F_{N,V} = \Eel (1+o(1))$$
for large $N$, and this is of order $N^2$. This estimate corresponds to the fact that 
$$ \frac{1}{N} \sum_{j= 1} ^N \delta_{\bx_j} \simeq \mueq$$
with overwhelming probability. Subtracting this well-known first order behavior, we shall chiefly be interested in the behavior of the ``correlation energy''
\begin{equation}\label{eq:corr ener}
\FC_{N,V} :=  F_{N,V} - \Eel.
\end{equation}
The Euler-Lagrange equation for~\eqref{eq:intro MFE} leads to 
$$ 
\mueq = \frac{1}{4\pi} \Delta V \1_\Sigma 
$$
for a set $\Sigma \subset \R^2$ called the droplet. We only consider the case where 
\begin{equation}\label{eq:intro mueq}
\mueq = \frac{1}{\pi} \1_\Sigma
\end{equation}
so that we deal with a system whose density is to leading order flat on the droplet. The latter can however be multiply connected, and this shall be our chief concern. In this particular case, the Zabrodin-Wiegman prediction~\cite{ZabWie-06} (adjusted to take multiple-connectedness into account~\cite{JanManPis-94}) reads
\begin{multline}\label{eq:intro ZabWieg}
\FC_{N,V} = -\frac{1}{4} N \log N - \frac{1}{2}\left( \frac{\log 2\pi}{2} - 1 \right) N - \frac{6-\chi}{24} \log N \\ - \frac{\log (2\pi)}{4} - \chi \frac{\zeta' (-1)}{2} +\frac{1}{4} \log \mathrm{det} _\zeta (\Delta_{\R^2 \setminus \Sigma}) + o_N (1). 
\end{multline}
We refer to~\cite[Section~9.3]{Serfaty-24} or~\cite[Section~5.3]{ForByu-25} for a more detailed account. 
%The leading mean-field term, forcing the charge distribution to follow $\mueq$ has already been subtracted. As regards the rest of the expansion:
%The next terms are derived by looking at fluctuation around this overall distribution, at the microscopic interparticle distance. 

\medskip

\noindent$\bullet$ The $O(N \log N)$ term comes about because a Coulomb self-energy of each individual charge, cut-off at the natural inter-particle distance $\sim N^{-1/2}$ arises when zooming in. This leads to an energy $N \log \left(N^{-1/2} \right)$, to be multiplied by the temperature factor $1/2$ from~\eqref{eq:intro free ener}.  

\smallskip

\noindent$\bullet$ The $O(N)$ term is dictated by local density approximation. It can be recovered from an integral over $\bx \in \Sigma$ of the free-energy density of a jellium at density $\mueq (\bx)$. For a constant density, and at temperature $\beta =2$ (the Ginibre case), this leads to the claimed expression. Minimizing this term is what gives rise to the gaussian free field fluctuations. This can be guessed~\cite{JanManPis-94} by writing an electrostatic energy in terms of the potential $\phi$, the field $\nabla \phi$ and the charge distribution $-\Delta \phi$ (according to Laplace's equation)
$$ -\int_{\R^2} \phi \Delta \phi = \int_{\R^2} \left|\nabla \phi\right|^2$$
and replacing the usual partition function expressed in terms of charge density by a (formal) functional integral 
\begin{equation}\label{eq:FFintegral}
 \int e^{- \int_{\R^2} \left|\nabla \phi\right|^2} D\phi. 
\end{equation}

\smallskip

\noindent$\bullet$ The $\log N$ term has a purely topological origin, in that its prefactor involves only the Euler characteristic of the droplet
$$ \chi := 2 - b = 1 - n$$
where $b$ is the number of boundaries\footnote{For systems on surfaces, the number of handles is also involved.}, $n$ the number of holes, and the equality holds for a connected droplet (hence, a single outer boundary) that we shall restrict to shortly. Noteworthily, the occurrence of such a term in the expansion was conjectured~\cite{JanManPis-94} in analogy with the gaussian free field~\cite{CarPes-88}. Similar terms occur in spectral invariants of the Laplacian on a domain\footnote{``One can hear the number of holes in a drum'', see e.g.~\cite{KeaSin-67} and references therein}, naturally connected to the formal integral~\eqref{eq:FFintegral}.

\smallskip

\noindent$\bullet$ Amongst the $O(1)$ terms, another topological one involving $\chi$ occurs (with the derivative of the Riemann $\zeta$ function as prefactor), but the most interesting is the ($\zeta$-regularized) spectral determinant of the Laplacian in the \emph{exterior} of $\Sigma$, connected to~\eqref{eq:FFintegral}, which is formally the product of Laplacian eigenvalues.

\medskip

Some interesting terms are absent of the above expansion: for a multi-component droplet there are extra oscillatory terms~\cite{AmeChaCro-23,ByuKanSeo-23,AllForLahShe-25,AllLah-25,ByuLeeYan-21,Byun-25,Charlier-24}, and, for other values of the inverse temperature $\beta$ there is a $O\left(N^{1/2}\right)$ term corresponding to a contribution of the droplet's outer boundary. That this term vanishes at $\beta = 2$ is a remarkable prediction of~\cite{CanForTelWie-15}.

In this paper we are particularly interested in getting indications of the topological $\log N$ terms. We cannot however expand directly the free energy with the desired precision, even for the particular model we will define shortly. To make some progress, we instead observe some remarkable consequences of Conjecture~\eqref{eq:intro ZabWieg}.

Let external potentials $V_{1\to n}$ and $V_n$ be chosen so that the corresponding droplets are
\begin{align}\label{eq:intro droplet}
\Sigma_{1\to n} &= D(0,R_{1\to n}) \setminus \bigcup_{k=1} ^n H_k  \nonumber \\
\Sigma_{n} &= D(0,R_n) \setminus H_n
\end{align}
where $D(0,R)$ is the disk of center $0$ and radius $R$ and $H_k,k=1\ldots n$ are $n$ holes puncturing it. Since the total charge is fixed in~\eqref{eq:intro mueq} it must be that
\begin{align}\label{eq:area}
R_{1\to n} &= \sqrt{1 + \pi ^{-1} \sum_{k=1}^n |H_k|} \nonumber\\
R_n &=\sqrt{1 + \pi ^{-1} |H_n|}
\end{align}
Then we should have 
\begin{align}\label{eq:spec det holes}
\log \mathrm{det} _\zeta (\Delta_{\R^2 \setminus \Sigma_{1\to n}}) &= \log \mathrm{det} _\zeta (\Delta_{\R^2 \setminus D(0,R_{1\to n})}) + \sum_{k=1}^n \log \mathrm{det} _\zeta (\Delta_{H_k})\nonumber \\
 &= \frac{1}{3} \log R_{1\to n} + \sum_{k=1}^n \log \mathrm{det} _\zeta (\Delta_{H_k})
\end{align}
where the expression of the contribution of the exterior of $D(0,R_{1\to n})$ is taken from~\cite[Section~6.1]{ZabWie-06} (see also~\cite[Remark~2.3]{ByuSeoYan-24}) and the contributions from the holes are, by translation invariance of the GFF, independent of the locations of the holes within the droplet. From~\eqref{eq:intro ZabWieg} we infer that 

\medskip

\noindent(i) $\FC_{N,V_n}$ is, up to order $o_N(1)$, independent of the location of the hole $H_n$, as long as it stays away from the boundary of  $D(0,R_n)$.

\smallskip

\noindent(ii) The change in correlation energy when adding a hole in the droplet is 
 \begin{multline}\label{eq:ZabWiegplus1} 
\FC_{N,V_{1\to n}} - \FC_{N,V_{1\to n-1}} - \FC_{N,V_{n}} =  \frac{N \log N}{4} + \frac{1}{2}\left( \frac{\log 2\pi}{2} - 1 \right) N \\ + \frac{5\log N}{24} + \frac{\zeta'(-1)}{2} + \frac{1}{12} \log \frac{R_{1\to n}}{R_{1\to n-1} R_n} + o_N (1)
\end{multline}
% 
% \smallskip
% 
% \noindent(iii) In particular, if we \emph{assume} an expansion 
% \begin{align*}
%  \FC_{N,V_{1\to n}} &= -\frac{1}{4} N \log N - \frac{1}{2}\left( \frac{\log 2\pi}{2} - 1 \right) N + c_{1\to n} \log N +o (\log N)\\
%  \FC_{N,V_{n}}&= -\frac{1}{4} N \log N - \frac{1}{2}\left( \frac{\log 2\pi}{2} - 1 \right) N + c_{n} \log N + o (\log N)
% \end{align*}
% i.e. that the leading correction after the rigorously known terms is of order $\log N$, and if all holes are identically shaped, then it follows that 
% $$ c_{1\to n} = n c_n + \frac{5(n-1)}{24}$$
% hence the $\log N$ is indeed topological in nature. Further assuming that $c_n = -1/4$, in analogy with what is rigorously known for some particular models~\cite{ByuSeoYan-24}, then it must be that 
% $$c_{1\to n} = - \frac{n-5}{24} = - \frac{6-\chi}{24}$$ 
% as expected. 

\medskip
 
These are the consequences of~\eqref{eq:intro ZabWieg} that we manage to prove, in a particular model with sufficiently small and separated holes. We punch the holes in the droplet following~\cite{RouYng-17} by filling them with a suitable distribution of $M$ unit pinned charges. Our potential $V$ is the sum of a quadratic $|\bx|^2$ term (corresponding to a neutralizing ``jellium'' background and setting the constant value of the density in~\eqref{eq:intro mueq}) and the Coulomb potential generated by these pinned charges. 

The model we obtain this way benefits from a very useful exact formula~\cite{AkeVer-03,Lambert-20,LamLunRou-22}: its free-energy is proportional to the reduced $M$ particles density of the Ginibre ensemble (i.e. the same model, but without pinned charges) with $N+M$ particles, evaluated at the locations of the pinned charges. In this representation the properties above translate to

\medskip

\noindent(i) Said reduced $M$-particles density is, to the desired precision, translation-invariant. This we prove by controlling the error made by replacing, in suitable determinantal expressions, the finite $N+M$ Ginibre correlation kernel by the limiting correlation kernel of the Ginibre process.

\smallskip

\noindent(ii) If the $M$ pinned charges are split in two well-separated groups of $M_1$ and $M_2$ charges (with $M=M_1 + M_2$), the reduced $M$-particles density factorizes (clustering due to the fast decay of the Ginibre correlation kernel) into the individual contributions of the two groups, involving the reduced $M_1$-particles and $M_2$-particles densities, respectively.

\medskip

With a bit of book-keeping in exact formulas and a detailed analysis of the mean-field problem, these properties yield the desired variations of free energy for our conditioned Ginibre ensemble. In the proofs of  both properties, the main difficulty is to obtain reliable estimates with large $M\propto N$, for this is necessary to punch macroscopic holes in the droplet, and thus set the problem in the regime of conjectured applicability of~\eqref{eq:intro ZabWieg}.

\bigskip 

\noindent\textbf{Acknowledgments:} This work benefited from insightful conversations with Alice Guionnet, Gaultier Lambert, Thomas Lebl\'e and Sylvia Serfaty.

\section{Model and results}

We turn to a precise description of our model, the assumptions corresponding to our previous vague statements, and our main results.

\subsection{Pinned charge configuration}\label{sec:asum}

In essence we need the pinned charges to be ``evenly distributed in several sufficiently small and separated clusters''. Since we are defining a particular model on which to check some consequences of~\eqref{eq:intro ZabWieg}, we do not aim at over-optimizing the conditions below.

For two measures $\mu,\nu$ we define their Coulomb interaction energy
\begin{equation}\label{eq:Coul int}
D(\mu,\nu) = - \iint_{\R^2 \times \R^2} \mu(\bx) \log |\bx-\by| \nu (\by) d\bx d\by.  
\end{equation}
For $n\in \N$ and $j=1 \ldots n$ let $\left(\bw_{j,k}\right)_{k= 1 \ldots M_j}$ be $n$ sets of points in the plane. We shall denote
\begin{equation}\label{eq:fraction}
c_j = \frac{M_j}{N}, \quad M =  \sum_{j=1} ^n M_j, \quad c = \sum_{j=1} ^n c_j 
\end{equation}
and assume that each $M_j$ is of order $N$, so that $c_j$ is of order $1$ when $N\to \infty$. One of our key assumptions will be that $c$ is a small enough constant.

The following notion will be useful

\begin{definition}[\textbf{Screening region}]\label{def:screening}\mbox{}\\
We say that $H\subset\R^2$ is a \emph{screening region} for a set of points $\bw_k\in \R^2, k= 1 \dots M$ if 
\begin{equation}\label{eq:screening}
-\log |\,.\,|\star \left(\frac{1}{\pi} \1_{H}- \frac{1}{N}\sum_{k=1} ^M \delta_{\bw_k} \right) \begin{cases}
                                                                                                  = 0 \mbox{ on } H^c \\
                                                                                                  \leq 0 \mbox{ on } H
                                                                                                 \end{cases}
\end{equation}
where $\star$ stands for convolution.

In particular, it must be that 
\begin{equation}\label{eq:screening area}
|H| = \pi \frac{M}{N}
\end{equation}
and that $\bw_k \in H$ for all $k=1\ldots M$.
\end{definition}

Existence of the above region follows from the arguments in~\cite[Section~3]{LieRouYng-17}. Uniqueness was not considered therein but, under our assumptions below, it actually follows from Theorem~\ref{thm:MFE}. Screening regions are also known as subharmonic quadrature domains~\cite{GusSha-95,GusPut-07,Sakai-82}, see the discussion in~\cite[Remark~5.4]{Rougerie-Elliott} for further references. We will show in Section~\ref{sec:MF} below that the screening regions of the charge clusters correspond to the holes in the droplet.

\begin{assumption}[\textbf{Each cluster of charges evenly fills its screening region}]\label{asum:onehole}\mbox{}\\
For all $j=1\ldots n$, denote $H_j$ the screening region that Definition~\ref{def:screening} associates with the set of points $\left(\bw_{j,k}\right)_{k= 1 \ldots M_j}$. We demand  

\smallskip

\noindent \emph{(i) \underline{separation of charges}}. For fixed constants $C_1,C_2 >0$
\begin{equation}\label{eq:ref space}
C_1 M^{-1/2} \leq |\bw_{j,k} - \bw_{j,k'}| \leq C_2 M^{-1/2}. 
\end{equation}
where $\bw_{j,k'}$ is the nearest neighbor of $\bw_{j,k}$ within $\left(\bw_{j,k}\right)_{k= 1 \ldots M_j}$.

\smallskip

\noindent \emph{(ii) \underline{even distribution of energy}}. For large $N$
\begin{align}\label{eq:ref energy}
 \cH_{N} \left(\bw_{j,1}, \ldots, \bw_{j,M_j}\right) &:= \frac{N}{2} \sum_{k=1} ^{M_j} |\bw_{j,k}| ^2 - \sum_{1\leq k < \ell \leq M_j} \log |\bw_{j,k} - \bw_{j,\ell}| \nonumber \\
 &= \frac{N^2}{2\pi} \int_{H_j} |\bx|^2 d\bx + \frac{N ^2}{2\pi ^2} D \left(\1_{H_j},\1_{H_j}\right) - \frac{1}{2} M_j \log M_j + O (M)
\end{align}
where $|O (M)| \leq C M$ for a fixed constant $C>0$.
\end{assumption}

Item (i) ensures that we may always think of the pinned charges as individual ones. As for Item (ii), it implies that the empirical density 
\begin{equation}\label{eq:explain asum}
\rho_{M_j} ^{(1)}:= \sum_{k=1} ^{M_j} \delta_{\bw_{j,k}} \simeq \frac{N}{\pi} \1_{H_j} 
\end{equation}
in the sense of Coulomb energies (self-interaction plus interaction with a harmonic external potential). The local value of the density is the equilibrium one for the minimization of 
$$ \frac{N}{2} \int_{\R^2} |\bx|^2 \rho (\bx) + \frac{1}{2} D (\rho,\rho)$$
and hence the density of points we choose is at equilibrium with/screens a harmonic background potential in $H_j$.

We assume a matching of Coulomb energies only up to order $N\log N$, which fits squarely within the range of known estimates: recall that~\eqref{eq:intro ZabWieg} is known rigorously up to order $N$ for all $\beta$,  including $\beta = \infty$. The existing technology suffices to show that, for example, a regular lattice filling $H_j$ will satisfy both assumptions (see also Remark~\ref{rem:holes} below). At the level of precision demanded in~\eqref{eq:ref energy}, the apparent cyclicity in first defining a screening region associated to the charges, and then assuming that the latter fill it evenly, will not be a concern. For example, if one aims at a roughly disk-shaped $H_j$, a ground state configuration for $\cH_{N}$, suitably translated, will also satisfy our assumptions. 

Next we turn to

\begin{assumption}[\textbf{Clusters of charges are well-separated}]\label{asum:multhole}\mbox{}\\
For all $j=1\ldots n$, with the same notation as above, we demand that there be a disk $D_j$ of radius $R_j$ such that 
\begin{equation}\label{eq:disks}
H_j \subset D_j \mbox{ and } \bw_{j,k} \in D_j \,\mbox{ for all } k= 1 \ldots M_j.
\end{equation}
We impose 
\begin{equation}\label{eq:asumscales}
R_j \leq r_1 \min_{j\neq j'} \mathrm{dist} \left(D_j, D_{j'}\right) \leq r_1 r_2  \min_{j} \mathrm{dist} \left(D_j, \partial D(0,R_n)\right)
\end{equation}
with $r_1,r_2$ two sufficiently small constants and
\begin{equation}\label{eq:radius drop}
R_n := \sqrt{1 + \sum_{j=1}^n c_j}.
\end{equation}
\end{assumption}

From~\eqref{eq:screening area} we have that 
$$ |H_j| = \pi c_j$$
and thus, for disjoint holes, $R_n$ above is the outer radius of the droplet, ensuring a fixed total charge:
$$ \frac{1}{\pi} \left| D(0,R_n) \setminus \bigcup _j H_j \right| = 1.$$ The above assumptions thus mean that the size of the holes must be sufficiently smaller than their mutual distance, which itself must be sufficiently smaller than their distance to the droplet's outer boundary: 
 \begin{equation}\label{eq:ineq points}
 \max_{j,k\neq k'} |\bw_{j,k} - \bw_{j,k'}| \leq r_1 \min_{j\neq j',k,k'}  |\bw_{j,k} - \bw_{j,k'}| \leq r_1 r_2 \min_{j,k} \mathrm{dist} \left(\bw_{j,k},\partial D(0,R_n)\right).
\end{equation}

\subsection{Main results} 

We come to our results, vindicating the consequences of Conjecture~\eqref{eq:intro ZabWieg} we have been discussing in the introduction, for the particular model we just defined. Namely, we look at the partition function appearing in~\eqref{eq:Gibbs}-~\eqref{eq:intro free ener} where the Hamiltonian~\eqref{eq:Hamil} is set as 
\begin{equation}\label{eq:choice pot}
V(\bx) := |\bx| ^2 - \frac{2}{N} \sum_{j=1} ^n \sum_{k=1} ^{M_j} \log \left|\bx - \left(\bw_{j,k} + \ba_j\right)\right|. 
\end{equation}
The first term is the usual quadratic potential for the Ginibre ensemble, leading to a flat local density. The second term is the Coulomb potential generated by several sets of pinned charges as described above. The vectors $\ba_1, \ldots, \ba_n$ are translations that can act on each of the cluster of pinned charges, to investigate the effect of moving holes around. For convenience we regard the reference sets of points $\left(\bw_{j,k}\right)_{k= 1 \ldots M_j}$ as fixed, and only vary the translation vectors $\ba_1, \ldots, \ba_n$. Our running assumption will always be that 
\begin{equation}\label{eq:runasum}
\boxed{\mbox{the $n$ point configurations } \left(\bw_{j,k} +\ba_j \right)_{k= 1 \ldots M_j} \mbox{ satisfy Assumptions~\ref{asum:onehole} and~\ref{asum:multhole}} }
\end{equation}
which can be achieved by asking that the assumptions are satisfied for $\ba_1, \ldots, \ba_n = 0$ and then only considering variations with $|\ba_j|$ small enough for all $j= 1 \ldots n$.

The partition functions we look at are thus in the form
\begin{multline}\label{eq:def part}
\cZ_N (\ba_1,\ldots,\ba_n) := \int_{\R^{2N}}\prod_{1\leq j < k \leq N} |z_j-z_k|^2 e^{-N\sum_{j=1}^N |z_j|^2} \\\times \prod_{\ell = 1} ^N \prod_{j=1} ^n \prod_{k=1}^{M_j} |z_\ell - w_{j,k} - a_j|^2 dz_1 \ldots dz_N 
\end{multline}
where we identify vectors $\bw_{j,k},\ba_j$ with complex numbers $w_{j,k},a_j$. Following the introductions this leads to the free energies and correlation energies
\begin{align}\label{eq:def free}
F_N (\ba_1,\ldots,\ba_n) &= -\frac{1}{2}\log \cZ_N (\ba_1,\ldots,\ba_n)\nonumber\\
\FC_N (\ba_1,\ldots,\ba_n)&= F_N (\ba_1,\ldots,\ba_n) - \Eel (\ba_1,\ldots,\ba_n)
\end{align}
where the mean-field energy $\Eel (\ba_1,\ldots,\ba_n)$ is defined by inserting~\eqref{eq:choice pot} in~\eqref{eq:intro MFE}.

Note that one may think of the above model as an enlarged Ginibre ensemble (no pinned charges, only quadratic external potential) of $N+M$ particles, conditioned on fixing $M$ particles as described above. The asymptotics we are interested in are also related to moments of the Ginibre ensemble's characteristic polynomial and Fisher-Hartwig asymptotics~\cite{AkeVer-03,BouDubHarKel-25,ChaGarXu-26,Lambert-20}

Our first result investigates the correlation energy $\FC_N (\ba)$ for a single hole/cluster of pinned charges. The prediction of~\eqref{eq:intro ZabWieg} in this case is that there is no dependence on $\ba$ up to order $o_N (1)$. Hence the only variations  of $F_N (\ba)$ occur at the macroscopic/mean-field level $N^2$ of the expansion, see Section~\ref{sec:MF} below. 

\begin{theorem}[\textbf{Moving a single hole around the droplet}]\label{thm:onehole}\mbox{}\\
Let $n=1$, i.e. pick $\bw_1,\ldots,\bw_M$ fixed points satisfying Assumption~\ref{asum:onehole} and set
$$ c = \frac{M}{N}$$
and/or $|\ba|$ small enough (which guarantees~\eqref{eq:asumscales} in this case). Then, with the above notation,
\begin{equation}\label{eq:resultonehole}
\left| \FC_N (\ba) - \FC_N (0)\right| \leq o_N (1)
%\Eel -\frac{1}{2} N \log N - \left( \frac{\log 2\pi}{2} - 1 \right) N \\ 
%- \frac{1}{2} \log N - \frac{\log (2\pi)}{2} - \frac{1}{12} \log\frac{c}{1+c} + o_N (1) 
\end{equation}
in the $N\to \infty$ limit, with $\left|o_N (1) \right|\leq e^{-C N}$. 
\end{theorem}

Although we explicitly consider only translations of the hole/cluster of points, note the following:

\begin{remark}[Rotating the hole]\label{rem:rot}\mbox{}\\
It is clear from~\eqref{eq:def part} that $\FC_N (0)$ is invariant under a joint rotation of $\bw_1,\ldots,\bw_M$ around the origin. Using the theorem above to translate an arbitrary rotation center to the origin, and back to its original location, one deduces that $\FC_N (\ba)$ is also, up to exponentially small remainders, invariant under a joint rotation of all the pinned charges around any center, as long as Assumption~\ref{asum:onehole} and~\eqref{eq:asumscales} hold all along the rotation.\hfill$\diamond$
\end{remark}

In spirit, Theorem~\ref{thm:onehole} is reminiscent of~\cite[Proposition~3.5, Item (i)]{ByuSeoYan-24}, which corresponds to the case where all pinned points are collapsed into a single one, leading to a disk-shaped hole. In as much as the two results can be compared, we work under much more restrictive assumptions on the total pinned charge and its location, but allow for an arbitrarily shaped hole. 

Next we turn to the case of multiple holes:

\begin{theorem}[\textbf{Punching multiple holes in the droplet}]\label{thm:multhole}\mbox{}\\
Pick $n$ configurations of points and $n$ translation vectors so that Assumptions~\ref{asum:onehole} and ~\ref{asum:multhole} hold for the translated point clusters $\bw_{j,k}+\ba_j,k=1 \ldots M_j$. Then 
\begin{align}\label{eq:resultmultipleeholes}
\FC_N (\ba_1, \ldots,\ba_n) &- \sum_{j=1} ^n \FC_N (\ba_j)= \frac{n-1}{4}N\log N + \frac{1}{2} \left( \frac{\log 2\pi}{2} - 1\right) (n-1) N \nonumber\\
&+ \frac{5(n-1)}{24} \log N + (n-1) \frac{\zeta' (-1)}{2} + (n-1) \frac{\log 2\pi}{4} \nonumber\\
&+  \frac{1}{24}\left(\log (1+c) - \sum_{j=1} ^n \log (1+c_j) \right)+ o_N (1)
\end{align}
where the charges $c$ and $c_j,j=1\ldots n$ are as in~\eqref{eq:fraction} and $\left|o_N (1) \right|\leq e^{-C N}$.
\end{theorem}

Combining with Theorem~\ref{thm:onehole} and Remark~\ref{rem:rot} shows that, at least as long as the holes are sufficiently small and separated, the free energy depends on their locations and relative orientations only through the mean-field term. Our estimate~\eqref{eq:resultmultipleeholes} is  an iterated version of~\eqref{eq:ZabWiegplus1}. As regards the last line, to compare with~\eqref{eq:area} and~\eqref{eq:spec det holes}, recall from~\eqref{eq:radius drop} that $\sqrt{1+c}$ and $\sqrt{1+c_j}$ are the outer radii of the droplets with all holes present, respectively only the $j$-th one.  

Let us clarify how the above is suggestive of the occurrence of topological terms in the plasma's free energy:

\begin{remark}[Topological terms]\label{rem:top}\mbox{}\\
For simplicity, assume that the $n$ clusters of points are identical, leading to $n$ similarly-shaped holes that we can translate and rotate inside the droplet, as long as~\eqref{eq:ineq points} stays valid. If we \emph{postulate} that the next term after the rigorously known $O(N)$ ones in the expansion of $\FC_N (\ba_1, \ldots,\ba_n)$ is indeed of order $\log N$, we get an expansion of the form
\begin{align*}
 \FC_N (\ba_1, \ldots,\ba_n) &= -\frac{1}{4} N \log N - \frac{1}{2}\left( \frac{\log 2\pi}{2} - 1 \right) N + c_{1\to n} \log N +o (\log N)\\
 \FC_N (\ba_1) &= -\frac{1}{4} N \log N - \frac{1}{2}\left( \frac{\log 2\pi}{2} - 1 \right) N + c_{1} \log N + o (\log N).
\end{align*}
Then it follows from Theorem~\ref{thm:onehole} that for all $j=1\ldots n$
$$\FC_N (\ba_j) = -\frac{1}{4} N \log N - \frac{1}{2}\left( \frac{\log 2\pi}{2} - 1 \right) N + c_{1} \log N + o (\log N)$$
and from Theorem~\ref{thm:multhole} that 
$$ c_{1\to n} =  \left( c_1 + \frac{5}{24}\right) n - \frac{5}{24}.$$
Hence the $\log N$ term must indeed be topological in nature. It might still be ``trivially topological'', i.e. universal, if it so happens that 
$$c_1 = -\frac{5}{24}. $$
In view of known results for models (different from ours) with a single hole (see aforementioned references, in particular~\cite{ByuSeoYan-24}), this seems quite unlikely. It is much more natural to expect that 
$$ c_1 = -\frac{1}{4}$$
leading to  
$$c_{1\to n} = - \frac{n+5}{24} = - \frac{6-\chi}{24}$$ 
as predicted by Conjecture~\eqref{eq:intro ZabWieg}, so that the $\log N$ term indeed counts the number of holes in the droplet. Similar considerations apply to other topological terms at level $O(1)$ in the expansion.\hfill $\diamond$
\end{remark}

\bigskip 

The rest of the paper, containing the proofs of Theorems~\ref{thm:onehole} and~\ref{thm:multhole}, is organized as follows:
\begin{itemize}
 \item In Section~\ref{sec:MF} we set up preliminary estimates on the mean-field approximation of the problem. These show that the variations we will later find in $F_N(\ba_1,\ldots,\ba_n)$ are all accounted for by those of the mean-field energy.
 \item In Section~\ref{sec:proofone} we prove Theorem~\ref{thm:onehole}. In particular, we recap the representation of the partition function in terms of a Ginibre correlation function. Our assumption~\eqref{eq:ref energy} implies useful a priori bounds on the later, that will enter all subsequent estimates. In particular when replacing finite area Ginibre correlation functions by infinite area ones, which is the next big task of the section. 
 \item In Section~\ref{sec:proofmult} we prove Theorem~\ref{thm:multhole}. Following on the representation just mentioned, this boils down to an exponential clustering estimate for Ginibre correlation functions, and a careful computation to identify constant terms in expansions. We rely heavily on the determinal structure for the clustering estimate.
 \item For the convenience of the reader, Appendix~\ref{app:Ginibre} recalls known facts about the Ginibre partition function and correlation functions.
\end{itemize}

\section{Mean-field considerations}\label{sec:MF}

Here we study the mean-field approximation of the model described above. In particular we investigate how the ground state energy depends on movements of a cluster of pinned charges and/or the addition of a cluster. This will be useful later, in comparison with the behavior of the full many-body problem, to reconstruct the desired behavior of the correlation energy.

Let $\bw_1,\ldots,\bw_M$ be $M$ points in the plane. We consider the mean-field energy functional 
\begin{multline}\label{eq:MFE bis}
\cEel [\sigma] := \frac{NJ}{2}\int_{\R^2} \left(|\bx|^2 - \frac{2}{N} \sum_{k=1}^M \log |\bx - \bw_k| \right) \sigma (\bx) d\bx \\
- \frac{J^2}{2}\iint_{\R^2 \times \R^2}\sigma(\bx) \log |\bx-\by| \sigma(\by) d\bx d\by
\end{multline}
for parameters $N>0,J>0,M\in \N$ and pinned charges $\bw_k\in \R^2, k= 1 \dots M$. The associated minimization problem is 
\begin{equation}\label{eq:MFE min}
\Eel = \inf \left\{ \cEel [\sigma], \sigma \in L^2 (\R^2), \sigma \geq 0, \int_{\R^2} \sigma = 1 \right\}.
\end{equation}
For simplicity we do not indicate the extra parameter $J$  in the notation (it was set as $J=N$ in the previous sections).  Varying $J$ will be helpful because we will need later to consider ensembles with the same background charge density (set by the real parameter $N$ in front of the $|\bx|^2$ term from~\eqref{eq:Hamil}-~\eqref{eq:choice pot}) but different particle numbers (set by the number $N$ of terms in the sums of~\eqref{eq:Hamil}).

%We will use freely the following concept
% 
% \begin{definition}[\textbf{Screening region}]\label{def:screening}\mbox{}\\
% We say that $H\subset\R^2$ is a \emph{screening region} for a set of points $\bw_k\in \R^2, k= 1 \dots M$ if 
% \begin{equation}\label{eq:screening}
% -\log |\,.\,|\star \left(\frac{1}{\pi} \1_{H}- \frac{1}{N}\sum_{k=1} ^M \delta_{\bw_k} \right) \begin{cases}
%                                                                                                   = 0 \mbox{ on } H^c \\
%                                                                                                   \leq 0 \mbox{ on } H^c.
%                                                                                                  \end{cases}
% \end{equation}
% In particular, it must be that 
% \begin{equation}\label{eq:screening area}
% |H| = \pi \frac{M}{N}
% \end{equation}
% and that $\bw_k \in H$ for all $k=1\ldots M$.
% \end{definition}
% 
% That the above definition is non empty follows from the arguments in~\cite[Section~3]{LieRouYng-17}. Screening regions are also known as subharmonic quadrature domains~\cite{GusSha-95,GusPut-07,Sakai-82}, see the discussion in~\cite[Remark~5.4]{Rougerie-Elliott} for further references.

Regarding the minimization of the mean-field energy~\eqref{eq:MFE bis} we will need the following result. In particular, observe that, although the uniqueness of a screening region was not discussed in Definition~\ref{def:screening}, it follows from Item (i) below. 

\begin{theorem}[\textbf{The mean-field problem}]\label{thm:MFE}\mbox{}\\
\noindent \emph{(i) \underline{equilibrium measure.}} Assume that the screening region $H$ associated to $\bw_k\in \R^2, k= 1 \dots M$ by Definition~\ref{def:screening} satisfies 
\begin{equation}\label{eq:screening subset}
H \subset D(0,R), \mbox{ with } R = \sqrt{\frac{J}{N} + \frac{M}{N}}.
\end{equation}
Then the unique solution $\mueq$ of~\eqref{eq:MFE min} is given by 
\begin{equation}\label{eq:MFsol}
\mueq = \frac{N}{\pi J} \1_{D(0,R) \setminus H}
\end{equation}
and the associated minimal energy is 
\begin{align}\label{eq:MFsol2}
\Eel &= \frac{1}{2} C_R - \frac{J^2}{2} D(\mueq,\mueq) \nonumber\\
C_R &= NJ R^2 - 2NJ R^2 \log R 
\end{align}
\noindent \emph{(ii) \underline{translating the pinned charges.}} Let $\ba \in \R^2$ and denote $\Eel(\ba)$ the minimal energy corresponding to the points $\bw_1+ \ba, \dots, \bw_M + \ba$. As long as~\eqref{eq:screening subset} holds for the associated screening region $H(\ba)$ we have that 
\begin{equation}\label{eq:grad MF}
\nabla_\ba \Eel(\ba) = - N \sum_{j=1} ^M \left( \bw_j + \ba \right)
\end{equation}
\noindent \emph{(iii) \underline{adding a cluster of pinned charges.}} Assume in addition that the points $\bw_k\in \R^2, k= 1 \dots M$ can be split into two groups of $M_1$ points $\bw_{1,j}, j=1 \ldots M_1$ and $M_2$ points $\bw_{2,j}, j=1 \ldots M_2$, with screening regions $H_1,H_2$ respectively. 

Assume that $H_1 \cap H_2 = \varnothing$. Let $\Eel_{12},\Eel_1,\Eel_2$ denote the infima of~\eqref{eq:MFE bis} with all the points taken into account, and with respectively only the points of the first or second group. Let correspondingly $R_{12},R_1,R_2,C_{R_{12}},C_{R_1},C_{R_2}$ be defined as above. Then 
\begin{align}\label{eq:split MFE}
 \Eel_{12} - \Eel_1 - \Eel_2 &= \frac{1}{2}\left(C_{R_{12}} -C_{R_{1}} - C_{R_{2}}\right) \nonumber\\
 &+ \sum_{j=1}^{M_1} \sum_{k=1}^{M_2} \log |\bw_{1,j} - \bw_{2,k}|\nonumber\\ 
 &-\frac{N^2}{2}\left(\frac{R_{12}^4}{4} -R_{12}^4 \log R_{12} - \frac{R_{1}^4}{4} + R_{1}^4 \log R_{1}- \frac{R_{2}^4}{4}+ R_{2}^4 \log R_{2}\right)\nonumber\\
 &-  M_1 N \left(R_{12}^2 \log R_{12} - \frac{R_{12}^2}{2} - R_{1}^2 \log R_{1} + \frac{R_{1}^2}{2}\right) \nonumber\\
 &- {M_2 N} \left(R_{12}^2 \log R_{12} - \frac{R_{12}^2}{2} - R_{2}^2 \log R_{2} + \frac{R_{2}^2}{2}\right).
\end{align}
\end{theorem}

\begin{remark}[Shaping the holes]\label{rem:holes}\mbox{}\\
It can be helpful to compare with~\cite{RouYng-17}, whose construction of droplets with arbitrary holes inspires the present one. In Proposition~3.1 and Section 3.4 therein it was proved that a droplet with arbitrary, fixed, holes (say a set $\widetilde{H}$) can be approximated by the minimizer of the mean-field problem with many individual charges pinned on a lattice filling the holes. In view of the above, this implies that, with $H$ the screening region of the pinned charges, $H \to \widetilde{H}$ when the lattice spacing goes to $0$ (in a topology and with a rate of convergence that we do not make precise for brevity). This shows how one can construct a droplet whose holes $H$ are close to any desired shape $\widetilde{H}$, by using only the potential generated by individual pinned charges. \hfill $\diamond$
\end{remark}

\begin{proof}[Proof of Theorem~\ref{thm:MFE}]
Existence and uniqueness of a minimizer $\mueq$ is standard for this convex functional, see e.g.~\cite[Chapter~1]{SafTot-97} or~\cite[Chapter~2]{Serfaty-24}. The Euler-Lagrange equation takes the form 
\begin{align}\label{eq:ELeq}
NJ |\bx|^2 -2 J \log |\, . \,| \star \left(J \mueq + \sum_{k=1} ^M \delta_{\bw_j} \right)&= C \mbox{ on } \mathrm{supp} (\mueq) \nonumber\\
NJ |\bx|^2 -2 J \log |\, . \,| \star \left(J \mueq + \sum_{k=1} ^M \delta_{\bw_j} \right)&\geq C \mbox{ on } \mathrm{supp} (\mueq) ^c
\end{align}
for a constant $C\in \R$ (Lagrange multiplier for the mass constraint). A useful characterization~\cite[Theorem~3.3, page 44]{SafTot-97} is that if~\eqref{eq:ELeq} holds for some probability measure $\mueq$ and some constant $C$, then $\mueq$ must be the unique minimizer. We thus argue that~\eqref{eq:MFsol} satisfies this, with $C= C_R$ as in~\eqref{eq:MFsol2}.  

First observe that~\eqref{eq:screening subset} and~\eqref{eq:screening area} imply that~\eqref{eq:MFsol} indeed is a probability measure. Next it follows from Newton's theorem (see~\cite[Theorem~9.7]{LieLos-01}) that
\begin{equation}\label{eq:Newton} 
-\frac{1}{\pi} \log |\, . \,| \star  \1_{D(0,R)} (\bx) = \begin{cases} - R^2 \log |\bx| \mbox{ for } |\bx| \geq R\\
                                       - \frac{|\bx|^2}{2} + \frac{R^2}{2} -R^2 \log R \mbox{ for } |\bx| \leq R.
                                      \end{cases}
\end{equation}
Combining with~\eqref{eq:screening} and observing that 
$$
NJ r^2 -2 NJ R^2 \log r \geq C_R \mbox{ for } r\geq R
$$
we find that 
\begin{multline*}
NJ |\bx|^2 -2 J \log |\, . \,| \star \left(J \mueq + \sum_{k=1} ^M \delta_{\bw_j} \right) \\= NJ \left(|\bx|^2 - \frac{2}{\pi} \log |\, . \,| \star \1_{D(0,R)}\right) - 2 J \log |\, . \,| \star \left(\sum_{k=1} ^M \delta_{\bw_j} - \frac{N}{\pi} \1_{H}\right)
\end{multline*}
indeed satisfies the desired conditions~\eqref{eq:ELeq}. Multiplying those by $\mueq$ and integrating we find the expression of the energy in~\eqref{eq:MFsol2}, thus concluding the proof of Item (i).

We turn to Item (ii). Let $\mueq^\ba$ be the equilibrium measure corresponding to the pinned charges at $\bw_1 (\ba) = \bw_1 + \ba, \dots, \bw_M (\ba) = \bw_M + \ba$. From~\eqref{eq:MFsol2} we have that 
$$ \nabla_\ba \Eel (\ba) = - \frac{N^2}{2 \pi^2} \nabla_\ba D\left( \1_{D(0,R) \setminus H(\ba)}, \1_{D(0,R) \setminus H(\ba)}\right).$$
Denote 
$$ \mathrm{Emp}^\ba := \frac{\pi}{N} \sum_{j=1}^M \delta_{\bw_j (\ba)} $$
and write 
\begin{align*}
D\left( \1_{D(0,R) \setminus H(\ba)}, \1_{D(0,R) \setminus H(\ba)}\right) &= D\left( \1_{D(0,R)}, \1_{D(0,R)}\right) + D\left( \1_{H(\ba)}, \1_{H(\ba)}\right) - 2 D\left( \1_{D(0,R) }, \1_{H(\ba)}\right) \\
&= D\left( \1_{D(0,R)}, \1_{D(0,R)}\right) + D\left( \1_{H(\ba)}, \1_{H(\ba)}\right) - 2 D\left( \1_{D(0,R) \setminus H(\ba)}, \1_{H(\ba)}\right) \\
&- 2 D\left( \1_{H(\ba)}, \1_{H(\ba)}\right) \\
&=  D\left( \1_{D(0,R)}, \1_{D(0,R)}\right) - D\left( \1_{H(\ba)}, \1_{H(\ba)}\right) - 2 D\left( \1_{D(0,R) \setminus H(\ba)},  \mathrm{Emp}^\ba \right)\\
&= D\left( \1_{D(0,R)}, \1_{D(0,R)}\right) - D\left( \1_{H(\ba)}, \1_{H(\ba)}\right) + 2 D\left( \1_{H(\ba)},  \mathrm{Emp}^\ba \right) \\
&- 2 D\left( \1_{D(0,R)},  \mathrm{Emp}^\ba \right)
\end{align*}
where we used~\eqref{eq:screening} to get the third equality. It follows from Definition~\ref{def:screening} that $H(\ba)$ is just $H(0)$ translated by $\ba$. Hence only the very last term of the right-hand side does depend on~$\ba$. Recalling~\eqref{eq:Newton} we find that 
\begin{align*}
\nabla_\ba D\left( \1_{D(0,R)},  \mathrm{Emp}^\ba \right) &= - \frac{\pi^2}{2N}  \nabla_\ba \left( \sum_{j=1}^M |\bw_j + \ba|^2\right) \\
&= - \frac{\pi^2}{N} \sum_{j=1} ^M \left( \bw_j + \ba \right).
\end{align*}
Combining with the two previous equations gives~\eqref{eq:grad MF}.

% The Feynman-Hellmann principle (discussed e.g. in~\cite[Lemma~2.8]{Rougerie-EMS}) gives
% $$ \nabla_\ba \Eel(\ba) = \left(\nabla_\ba \cEel (\ba) \right) \left[\mueq^\ba \right]$$
% where $\nabla_\ba \cEel (\ba)$ is the gradient of the functional itself and $\mueq^\ba$ the equilibrium measure corresponding to the pinned charges at $\bw_1 (\ba) = \bw_1- \ba, \dots, \bw_M (\ba) = \bw_M - \ba$. Thus 
% $$ \nabla_\ba \Eel(\ba) =  - J\int_{\R^2} \sum_{j=1}^M \frac{\bx-\bw_j (\ba)}{|\bx - \bw_j (\ba)|^2} \mueq^\ba (\bx) d\bx $$

As regards Item (iii), first note that since $H_1 \cap H_2 = \varnothing$ we have from Definition~\ref{def:screening} that $H= H_1 \cup H_2$ is a screening region for the total set of points $\bw_k\in \R^2, k= 1 \dots M$. Hence~\eqref{eq:MFsol2} and~\eqref{eq:MFsol} lead to 
\begin{align*} 
2 \Eel_{12} &= C_{R_{12}}- \frac{N^2}{\pi^2} \left( D\left(\1_{D(0,R_{12})},\1_{D(0,R_{12})}\right)  - 2 D\left(\1_{D(0,R_{12})},\1_{H_1} \right) - 2D\left(\1_{D(0,R_{12})},\1_{H_2} \right) \right) \\ 
&- \frac{N^2}{\pi^2} \left(D\left(\1_{H_1},\1_{H_1}\right) + D\left(\1_{H_2},\1_{H_2}\right) + 2 D\left(\1_{H_1},\1_{H_2}\right) \right) 
\end{align*}
with related expressions for $\Eel_{1}, \Eel_{2}$. Hence 
\begin{align}\label{eq:MFEdiff}
2\left(\Eel_{12} - \Eel_{1} - \Eel_{2}\right) &=  C_{R_{12}} - C_{R_{1}} - C_{R_{2}} \nonumber\\
&- \frac{N^2}{\pi^2} \left( D\left(\1_{D(0,R_{12})},\1_{D(0,R_{12})}\right) - D\left(\1_{D(0,R_{1})},\1_{D(0,R_{1})}\right) - D\left(\1_{D(0,R_{2})},\1_{D(0,R_{2})}\right)\right) \nonumber\\
&+2 \frac{N^2}{\pi^2} D\left(\1_{D(0,R_{12})\setminus D(0,R_{1})},\1_{H_1}\right) + 2 \frac{N^2}{\pi^2} D\left(\1_{D(0,R_{12})\setminus D(0,R_{2})},\1_{H_2}\right) \nonumber\\
&- 2 \frac{N^2}{\pi^2}D\left(\1_{H_1},\1_{H_2}\right)
\end{align}
Returning to~\eqref{eq:Newton} we have 
$$ 
D\left(\1_{D(0,R)},\1_{D(0,R)}\right) = \frac{\pi ^2 R^4}{4} - \pi^2 R^4 \log R.
$$
On the other hand, using Newton's theorem~\cite[Theorem~9.7]{LieLos-01} again implies that the Coulomb potential generated by $\1_{D(0,R_{12})\setminus D(0,R_{1})}$ is constant inside $D(0,R_{1})$, wherein $H_1$ is included. Hence, using~\eqref{eq:screening}, 
\begin{align*}
D\left(\1_{D(0,R_{12})\setminus D(0,R_{1})},\1_{H_1}\right) &= - \frac{\pi}{N} \sum_{j=1} ^{M_1} \log |\, . \,| \star \1_{D(0,R_{12})\setminus D(0,R_{1})} (\bw_{1,j}) \\
&= - \pi \frac{M_1}{N}  \log |\, . \,| \star \1_{D(0,R_{12})\setminus D(0,R_{1})} (0)\\
&= \pi ^2 \frac{M_1}{N} \left(R_{1}^2 \log R_{1} - \frac{R_{1}^2}{2} - R_{12}^2 \log R_{12} + \frac{R_{12}^2}{2}\right)
\end{align*}
and a similar expression with $R_1,H_1$ replaced by $R_2,H_2$. Since $H_1\cap H_2 = \varnothing$ it also follows from~\eqref{eq:screening} that 
$$ 
D\left(\1_{H_1},\1_{H_2}\right) = - \frac{\pi^2}{N^2} \sum_{j=1}^{M_1} \sum_{k=1}^{M_2} \log |\bw_{1,j} - \bw_{2,k}|.
$$
Combining the above calculations and inserting them in~\eqref{eq:MFEdiff} leads to~\eqref{eq:split MFE}.
\end{proof}

\section{Proofs in the one hole case}\label{sec:proofone}

Our general strategy for proving Theorem~\ref{thm:onehole} is as follows:

\medskip

\noindent $\bullet$ Since we are dealing with $M$ \emph{distinct} charges distributed around $\ba$, we can apply a simple exact formula for the corresponding partition function, originating in~\cite{AkeVer-03,Lambert-20} and used extensively in~\cite{LamLunRou-22}. This is based on the fact that our Gibbs state is a conditioned Ginibre ensemble.

\medskip 

\noindent $\bullet$ The formula gives~\eqref{eq:resultonehole} up to the logarithm of a determinant based on the finite $N$ Ginibre correlation kernel. Replacing the latter with the infinite area, translation invariant, correlation kernel, and controlling the error thus made,~\eqref{eq:resultonehole} follows suit. 

\medskip

\subsection{The exact formula}

% 
% \noindent $\bullet$ the analogue of~\eqref{eq:moveonehole} follows suit from the exact formula in the case with $M$ distinct charges. I.e. we can translate the tight cloud of point charges from $\ba$ to $0$, changing the corresponding partition function compatibly with~\eqref{eq:moveonehole}.
% 
% \medskip 
% 
% \noindent $\bullet$ once the $M$ point charges are distributed around $0$, we can contract them back to a single point of charge $2M=2cN$ with a negigible error, finishing the proof.

%\subsection{Exact formula with separated charges}

Let then 
\begin{equation}\label{eq:ref points}
\left(\bw_1,\ldots,\bw_{M}\right)\in \R^{2M}
\end{equation}
be a reference cloud of \emph{distinct} points. We assume~\eqref{eq:ref space} and~\eqref{eq:ref energy}.

We identify the vectors $\bw_1,\ldots,\bw_M$ with complex numbers $w_1,\ldots,w_M$ and $\ba$ with the complex number $a$. Define 
\begin{equation}\label{eq:partdiltran}
\ZN (\ba) := \int_{\R^{2N}}\prod_{1\leq j < k \leq N} |z_j - z_k| ^2 e^{-N\sum_{j=1}^N |z_j|^2} \prod_{j=1}^N \prod_{k= 1}^M \left|z_j- (w_k + a)\right|^{2} dz_1\ldots dz_N. 
\end{equation}
We shall use the a priori information that~\eqref{eq:intro ZabWieg} is already known rigorously up to $o_N (N)$:
\begin{equation}\label{eq:a priori}
F_N (\ba) = -\frac{1}{2} \log \ZN (\ba) = \Eel (N,N,M) - \frac{N}{4} \log N + \frac{N}{2} \beta f_2 (\beta) + o_N (N)
\end{equation}
where $\Eel (N,N,M)$ is the mean-field energy from section~\ref{sec:MF} at $J= N$ and $\beta f_2 (\beta)$ is the infinite area Jellium free-energy density, at inverse temperature $\beta = 2$, as defined in~\cite[Section~9 and references therein]{Serfaty-24}. We use the above at $\beta = 2$ where estimates for the Ginibre ensemble~\cite{ForByu-25} imply 
$$ 2 f_2 (2) = 2 \left( \frac{\log 2\pi}{2} - 1 \right).$$
The validity of~\eqref{eq:a priori} is usually investigated for a smooth, fixed external potential, not that generated by point charges that we consider. However, since the singularities generated by the point charges are outside of the droplet, a careful inspection of the known proofs shows that they carry over to our case.  In fact, our arguments below only require the direction $\leq$ of~\eqref{eq:a priori} which, as per~\eqref{eq:intro free func}, is the ``easy'' direction of the variational principle. Constructing a good trial state would be sufficient for our needs.

% Note that $\ZN (\ba,0)$ is just the partition function to be computed in Theorem~\ref{thm:onehole}. We want to obtain it from $\ZN (0,0)$ (given by~\eqref{eq:resultonehole0}) by differentiating in $\ba$, and dilating-contracting in $t$. The interest of the latter operation is that for $t>0$ we have that the points $t(\bw_j + \ba)$ are all distinct and hence  we can use the

%$$ N! \pi^{N+M} {\textstyle  \prod_{k=1}^{N+M-1} k! } \, N^{-(N+M)(N+M+1)/2}$$

We start our investigation of the remainder term in~\eqref{eq:a priori} by recalling an exact formula:

\begin{lemma}[\textbf{Exact expression for partition functions with pinned unit charges}] \label{lem:exp}\mbox{}\\
With the notation above and with $\ZGin_{N+M} = \cZ_{N+M} (\varnothing)$ the partition function with $N+M$ mobile charges and no pinned charge  (i.e. the Ginibre partition function~\eqref{eq:Ginibre def}), we have
\begin{multline}\label{eq:exact part}
\ZN (\ba) =  \ZGin_{N+M}  \frac{N!}{(N+M)!}
\det_{M\times M} \left[K_{N+M}(w_i + a,w_j +a) \right]  \frac{\prod_{j=1}^M e^{N|w_j +a|^2}}{\prod_{1\leq i <j \leq M} |w_i-w_j|^2} 
\end{multline}
where 
\begin{equation}\label{eq:KN}
K_{J} (z,w) = e ^{-\frac{N}{2}|z|^2 -\frac{N}{2}|w|^2} \sum_{j= 0}^{J-1} \frac{N^{j+1}}{\pi j!} z^j \overline{w}^j
\end{equation}
with the appropriate normalization is the Ginibre correlation kernel for $J$ particles in a background charge density $-4N$. 
\end{lemma}

\begin{proof}
This originates in~\cite{AkeVer-03,Lambert-20}, see for example~\cite[Appendix~A]{LamLunRou-22} for a proof of~\eqref{eq:exact part}. We used that for $M$ distinct points $w_1,\dots,w_M$
\begin{align}\label{eq:Ginibre marginal}
\frac{1}{M!}&\det_{M\times M} \left[K_{N+M}(w_i,w_j) \right] = \rho ^{(M)}_{N+M} (w_1,\ldots,w_M) := \frac{{N+M \choose M}}{\ZGin_{N+M}}\nonumber\\
&\int_{\R^{2N}}\prod_{1\leq j < k \leq N+M} |w_j - w_k| ^2 e^{-N\sum_{j=1}^{N+M} |w_j|^2} dw_{M+1} \ldots dw_{M+N}
\end{align}
the $M$-particles reduced density of a Ginibre ensemble with $N+M$ particles and correlation kernel $K_{N+M}$ as in~\eqref{eq:KN} (see Appendix~\ref{app:Ginibre}). 

\end{proof}

We will need some accurate estimates on the determinant appearing in~\eqref{eq:exact part}. This is to ensure that the errors we will later make by replacing it with the $N+M\to \infty$ version will indeed be negligible compared with its main contribution.

\begin{lemma}[\textbf{Lower bound on the determinant}]\label{lem:low det}\mbox{}\\
Under the previously stated assumptions, for a fixed positive constant $C>0$
\begin{equation}\label{eq:bound det 1}
 \det_{M\times M} \frac{\pi}{N} \left[K_{N+M}(w_i + a,w_j +a) \right] \geq \exp (- C \left( c - c\log c \right) N)   
\end{equation}
where $c = M/N.$
\end{lemma}

\begin{proof}
Starting from~\eqref{eq:exact part} and recalling the notation~\eqref{eq:ref energy} we find 
\begin{multline}\label{eq:pouet}
-\log \ZN (\ba) = -\log \ZGin_{N+M} - 2 \cH_N (\bw_1,\ldots,\bw_M) + \log \frac{(N+M)!}{N!}  - M \log \frac{N}{\pi} \\- \log \det_{M\times M} \frac{\pi}{N} \left[K_{N+M}(w_i + a,w_j +a) \right]. 
\end{multline}
From Stirling's formula we get 
\begin{align}\label{eq:AMN}
A(M,N) &= \log \frac{(M+N)!}{N!} - M \log \frac{N}{\pi} \nonumber
\\&= \frac{1}{2}\log \frac{N+M}{N}+(N+M) \log\frac{N+M}{N} + M(\log \pi - 1) +o_N (1)\nonumber 
\\&=(1+c) N \log (1+c) - cN \left( 1-\log\pi\right) + \frac{1}{2} \log (1+c) + o_N (1)
\end{align}
whereas asymptotics for the Ginibre ensemble recalled in~\eqref{eq:Ginibre part 3} lead to ($\beta = 2$)
\begin{multline}\label{eq:useGin} 
-\log \ZGin_{N+M} + 2 A (M,N) = \frac{3}{4} (1+c)^2 N^2  - \frac{(1+c)^2}{2} N^2 \log (1+c)  \\ - (1+c)\frac{N}{2} \log  N + (N+M)\beta f_2 (\beta) + c O(N)\\
= 2\Eel(N+M,N,0) - \frac{1+c}{2} N \log N + (N+M)\beta f_2 (\beta) + c O(N)
\end{multline}
where $\Eel(N+M,N,0)$ is the mean-field energy from Section~\ref{sec:MF} with $J=N+M,M= 0$ and $O(N)$ is bounded linearly in $N$.

Combining~\eqref{eq:a priori} with~\eqref{eq:pouet} and~\eqref{eq:useGin} and then inserting~\eqref{eq:ref energy} we find 
\begin{multline}\label{eq:quasi det bound} 
 \log \det_{M\times M} \frac{\pi}{N} \left[K_{N+M}(w_i + a,w_j +a) \right] \\ = 2\Eel(N+M,N,0) - 2\Eel(N,N,M) - \frac{N^2}{\pi} \int_{H} |\bx|^2 d\bx + \frac{N ^2}{\pi ^2} D \left(\1_{H},\1_{H}\right) \\ 
 + cN  \beta f_2 (\beta)  + cN \log c  -cN \left( 1- \log \pi\right) + c O (N)
\end{multline}
where $H$ is the screening region of the pinned charges. There now remains to observe that the terms on the second line cancel to conclude the proof. 

Indeed, with 
$$ M = cN \mbox{ and } R= \sqrt{1+c}$$
it follows from~\eqref{eq:MFsol} and~\eqref{eq:MFsol2} that 
\begin{align*} 
2\Eel(N+M,N,0) &= N(N+M)(1+c) - N (N+M)(1+c) \log (1+c) \\
&- \frac{N^2}{\pi^2} D\left(\1_{D(0,R)},\1_{D(0,R)}\right)\\
2\Eel(N,N,M)&= N^2 (1+c) - N^2 (1+c) \log (1+c) \\&- \frac{N^2}{\pi^2} D\left(\1_{D(0,R)},\1_{D(0,R)}\right) 
- \frac{N^2}{\pi^2} D\left(\1_{H},\1_{H}\right) + 2 \frac{N^2}{\pi^2} D\left(\1_{D(0,R)},\1_{H}\right)   
\end{align*}
and hence 
\begin{multline*}
 2\Eel(N+M,N,0) - 2\Eel(N,N,M) - \frac{N^2}{\pi} \int_{H} |\bx|^2 d\bx + \frac{N ^2}{\pi ^2} D \left(\1_{H},\1_{H}\right) \\
 = NM (1+c) - NM (1+c) \log (1+c) - 2 \frac{N^2}{\pi^2} D\left(\1_{D(0,R)},\1_{H}\right)  - \frac{N^2}{\pi} \int_{H} |\bx|^2 d\bx = 0
\end{multline*}
where we used~\eqref{eq:Newton} and~\eqref{eq:screening area} to compute $D\left(\1_{D(0,R)},\1_{H}\right)$ in the last step. 

Inserting in~\eqref{eq:quasi det bound} and exponentiating the resulting expression concludes the proof.
\end{proof}

\subsection{Moving the pinned charges}

We now use the exact formula from Lemma~\ref{lem:exp} to investigate the effect of a joint translation of the pinned charges. To this effect we first replace the correlation kernel $K_{N+M}$ by the corresponding, infinite area, kernel $K_\infty$. The error thus made is controlled thanks to Lemma~\ref{lem:low det}.

\begin{lemma}[\textbf{Inserting the translation-invariant kernel}]\label{lem:replace ker}\mbox{}\\
Let (with the usual identification $\R^2 \leftrightarrow \C$)
\begin{align}\label{eq:Kinfty}
K_{\infty} (z,w) &= e ^{-\frac{N}{2}|z|^2 -\frac{N}{2}|w|^2} \sum_{j= 0}^{\infty} \frac{N^{j+1}}{\pi j!} z^j \overline{w}^j\nonumber \\
&=\frac{N}{\pi} e^{-\frac{N}{2}\left(|z|^2 + |w|^2 - 2 z \overline{w} \right)}\nonumber \\
&= \frac{N}{\pi} e^{-\frac{N}{2}\left(|z-w|^2 - \i (\bz-\bw) \cdot (\bz + \bw)^{\perp} \right)}
\end{align}
and 
\begin{multline}\label{eq:Zinfty}
\ZN ^\infty (\ba) :=  \ZGin_{N+M} \frac{N!}{(N+M)!} \det_{M\times M} \left[K_{\infty}(w_i + a,w_j +a) \right]  \frac{\prod_{j=1}^M e^{N|w_j +a|^2}}{\prod_{1\leq i <j \leq M} |w_i-w_j|^2} 
\end{multline}
we have that, for $|\ba|,c$ small enough, 
\begin{equation}\label{eq:Zinfty2}
-\log  \ZN (\ba) = -\log  \ZN ^\infty (\ba) + o_N(1),
\end{equation}
with $o_N (1)$ exponentially small in the limit $N\to \infty$. 
\end{lemma}

\begin{proof}
Comparing~\eqref{eq:exact part} with~\eqref{eq:Zinfty}, writing 
\begin{multline*}
 \log \det_{M\times M}\left[\frac{\pi}{N}K_{\infty}(w_i + a,w_j +a) \right] - \log \det_{M\times M}\left[\frac{\pi}{N}K_{N+M}(w_i + a,w_j +a) \right] \\ 
 = \log \left(1 + \frac{\det_{M\times M}\left[\frac{\pi}{N}K_{\infty}(w_i + a,w_j +a) \right]-\det_{M\times M}\left[\frac{\pi}{N}K_{N+M}(w_i + a,w_j +a) \right]}{\det_{M\times M}\left[\frac{\pi}{N}K_{N+M}(w_i + a,w_j +a) \right]}\right) 
\end{multline*}
we need to prove that  
\begin{multline*}
 \left|\det_{M\times M}\left[\frac{\pi}{N}K_{N+M}(w_i + a,w_j +a) \right]-\det_{M\times M}\left[\frac{\pi}{N}K_{\infty}(w_i + a,w_j +a) \right]\right|\\ 
 \ll \det_{M\times M}\left[\frac{\pi}{N}K_{N+M}(w_i + a,w_j +a) \right]
 \end{multline*}
in the limit $N\to \infty$. In view of Lemma~\ref{lem:low det} it suffices to prove that
\begin{equation}\label{eq:bound det 2}
 \left|\det_{M\times M}\left[\frac{\pi}{N}K_{N+M}(w_i + a,w_j +a) \right]-\det_{M\times M}\left[\frac{\pi}{N}K_{\infty}(w_i + a,w_j +a) \right]\right| \leq e^{-D N} 
\end{equation}
for some fixed $D>0$, and then use the fact that $c=M/N$ is assumed small enough. 

Our bound~\eqref{eq:bound det 2} is obtained with an argument inspired by~\cite[Proof of Lemma~3.4.2]{AndGuiZei-10}. Let $H_k$ be the square $M\times M$ matrix 
\begin{itemize}
 \item whose $k-1$ first columns are the vectors 
 $$v_{kj} := \left(\frac{\pi}{N} K_{\infty}(w_i + a,w_j +a)\right)_{i=1,\ldots,M}$$
 for $j=1 \ldots k-1$
\item whose $k$-th column is the vector 
$$v_{kk} := \left(\frac{\pi}{N} K_{N+M}(w_i + a,w_k +a) - \frac{\pi}{N}K_{\infty}(w_i + a,w_k +a) \right)_{i=1,\ldots,M}$$
\item whose $M-k$ last columns are the vectors
$$v_{kj} := \left(\frac{\pi}{N} K_{N+M}(w_i + a,w_j +a)\right)_{i=1,\ldots,M}$$
 for $j=k+1 \ldots M$.
\end{itemize}
By linearity of the determinant with respect to columns we have 
\begin{equation}\label{eq:lin det}
\det_{M\times M}\left[\frac{\pi}{N}K_{N+M}(w_i + a,w_j +a) \right]-\det_{M\times M}\left[\frac{\pi}{N} K_{\infty}(w_i + a,w_j +a) \right] = \sum_{k=1} ^M \det_{M\times M} H_k 
\end{equation}
and by Hadamard's inequality 
\begin{equation}\label{eq:Hadamard}
 \left| \det_{M\times M} H_k\right| \leq \prod_{j=1} ^M \left(\sum_{i=1} ^M |v_{kj}^i|^2 \right)^{1/2} 
\end{equation}
with $v_{kj}^i$ the $i$-th element of the vector $v_{kj}$. We will bound the above terms using the estimates on correlation kernels recalled in Appendix~\ref{app:Ginibre}. To this end, note that~\eqref{eq:asumscales} with $r_2$ small enough and a choice of $|\ba|$ small enough imply that 
$$ |w_j+a| \leq \sqrt{1 +c} - \delta$$
for some $\delta >0$, so that we may in particular use~\eqref{eq:Kdiff} to obtain 
\begin{equation}\label{eq:Kdiff main}
 \left| K_{N+M}(w_i + a,w_j +a)-K_{\infty}(w_i + a,w_j +a)\right| \leq C e^{-C\delta N} 
\end{equation}
for all $i,j$.

Hence, using~\eqref{eq:Kmodule} and~\eqref{eq:Kdiff} we have, for $j\neq k$ 
\begin{equation}\label{eq:use Hadamard}
 \sum_{i=1} ^M |v_{kj}^i|^2  \leq C \sum_{i=1} ^M \left( e^{-N|w_i - w_j| ^2} + C e^{-CN}\right). 
 %\leq C N^{-1} \int_{\R^d} e^{-C N|\bx|^2} d\bx + C \leq C 
\end{equation}
But, in view of our choice of configuration $\bw_1,\ldots,\bw_M$, in particular~\eqref{eq:ref space}, the points can be sorted into clusters whose distance to a given $w_j$ is between $LN^{-1/2}$ and $(L+1)N^{-1/2}$, for integers $L$. The number of points in the $L$-th cluster cannot exceed $CL$ for some fixed constant $C$, and drops to $0$ for $L\geq C N^{1/2}$. Hence for $j\neq k$ 
\begin{equation}\label{eq:use Hadamard 2}
 \sum_{i=1} ^M |v_{kj}^i|^2  \leq C \sum_{L=0} ^{C\sqrt{N}} CL \left( e^{-C L ^2} + C e^{-CN}\right) \leq C. 
 %\leq C N^{-1} \int_{\R^d} e^{-C N|\bx|^2} d\bx + C \leq C 
\end{equation}
On the other hand~\eqref{eq:Kdiff main} gives, for $j=k$ 
\begin{equation}\label{eq:main gain}
\sum_{i=1} ^M |v_{kk}^i|^2 \leq M e^{-C \delta N}
\end{equation}
%where $C$ depends on the distance from the boundary of the droplet of the Ginibre ensemble for $N+M$ particles. 
Hence, combining~\eqref{eq:lin det} and~\eqref{eq:Hadamard} with~\eqref{eq:use Hadamard 2} and~\eqref{eq:main gain} we obtain a bound for the left-hand side of~\eqref{eq:bound det 2} of the order $M^{3/2} C^M e^{-C\delta N}$. Recalling that $M= cN$ and that $\delta$ can be bounded below by a fixed positive constant for $c,|\ba|$ small enough yields the desired~\eqref{eq:bound det 2}.
\end{proof}

We now use translation-invariance of the Ginibre process (whose correlation kernel is~$K_\infty$) to compute the gradient of the modified partition function~\eqref{eq:Zinfty}: 

\begin{lemma}[\textbf{Translation of the pinned charges}]\label{lem:trans}\mbox{}\\
With $\ZN ^\infty (\ba)$ as in~\eqref{eq:Zinfty} we have that 
\begin{equation}\label{eq:moveonehole2}
\nabla_\ba \log \ZN ^\infty (\ba)= 2 c \ba N^2 + 2 N \sum_{j=1} ^M \bw_j  
\end{equation}
\end{lemma}

\begin{proof}
We use that the $\log$ of~\eqref{eq:Zinfty} is the sum of several terms, only two of which do depend on $\ba$. In particular, the Vandermonde determinant in the denominator gives no contribution. 

We have
$$ \sum_{j=1}^M |\bw_j + \ba| ^2 = M |\ba|^2 + 2 \ba \cdot \sum_{j=1}^M \bw_j +  \sum_{j=1}^M |\bw_j|^2
%M |\ba|^2 + \sum_{j=1}^M |\bw_j|^2
$$
and hence, recalling~\eqref{eq:fraction},  
$$ \nabla_\ba \log \ZN ^\infty (\ba) = 2c N^2 \ba + 2 N \sum_{j=1} ^M \bw_j- \nabla_\ba \log \det_{M\times M} \left[K_{\infty}(w_i +a,w_j +a) \right]$$ 
and there remains to observe that $\det_{M\times M} \left[K_{\infty}(w_i +a,w_j +a) \right]$ does not depend on $\ba$ either. Indeed, according to~\eqref{eq:Ginibre marginal} and~\eqref{eq:Kinfty}, it is proportional to the $M-$particles density of a translation-invariant point process (the Ginibre point process on the full plane). More precisely, using the third formula in~\eqref{eq:Kinfty}
$$ \nabla_\ba K_{\infty}(w_i +a,w_j +a) = -\i N \ba^\perp \cdot (\bw_i - \bw_j) K_{\infty}(w_i +a,w_j +a) $$ 
and hence, expanding the determinant, 
\begin{multline*}
 \nabla_\ba \det_{M\times M} \left[K_{\infty}(w_i +a,w_j +a)\right]  \\ = -\i N \ba^\perp \cdot \sum_{\sigma \in \S_M} \mathrm{sgn}(\sigma) \sum_{j=1} ^M (\bw_j - \bw_{\sigma(j)})\prod_{i=1}^M K_{\infty}\big(w_i +a,w_{\sigma(i)}+a\big) = 0
 \end{multline*} 
because certainly
$$\sum_{j=1} ^M \bw_j - \sum_{j=1} ^M\bw_{\sigma(j)} = 0$$ 
for any permutation. This concludes the proof.
\end{proof}

We may now conclude the proof of Theorem~\ref{thm:onehole}. The argument is similar to ideas of~\cite{ByuSeoYan-24}.

\begin{proof}[Proof of Theorem~\ref{thm:onehole}]
Starting from Lemma~\ref{lem:replace ker}, we have that, under the stated assumptions and with an exponentially small $o_N (1)$
$$ \FC_N (\ba) = - \frac{1}{2} \log \ZN ^\infty (\ba) - \Eel (\ba) + o_N (1)$$
and 
$$ \FC_N (\ba) - \FC_N (0) = - \frac{1}{2} \log \ZN ^\infty (\ba) - \Eel (\ba) + \frac{1}{2} \log  \ZN ^\infty (0) + \Eel (0) + o_N (1).$$
But, combining~\eqref{eq:grad MF} (with $J=N$) and~\eqref{eq:moveonehole2} we conclude that the map  
$$ \ba \mapsto - \frac{1}{2} \log \ZN ^\infty (\ba) - \Eel (\ba)$$
is constant, and thus complete the proof.
\end{proof}

\section{Proofs for multiple holes}\label{sec:proofmult}

The main technical input in the proof of Theorem~\ref{thm:multhole} is a decoupling lemma for the determinant obtained by applying Lemma~\ref{lem:exp} to $\cZ_N(\ba_1,\ldots,\ba_n)$. We show that the main contribution is the product of the determinants obtained from applying the lemma to $\cZ_N(\ba_j)$ for $j=1,\ldots,\ba_n$. This is certainly intuitive: the multiple-holes-configuration's total determinant is made of diagonal blocks corresponding to each of the one-hole determinants, complemented with off-diagonal blocks whose fast decay can be controlled via the estimates recalled in Appendix~\ref{app:Ginibre}. This is a clustering property for correlation functions of a Ginibre ensemble when their arguments are sufficiently separated in space. 

The rest of the proof follows by inserting the exact formulae for Ginibre partition functions that we recall in Appendix~\ref{app:Ginibre} and comparing with the properties of the mean-field problem discussed in Section~\ref{sec:MF}.

\subsection{Decoupling the large determinant}

We state the decoupling lemma directly for the infinite Ginibre ensemble, replacing finite-$N$ correlation kernels by $K_\infty$.

\begin{lemma}[\textbf{Decoupling the multiple-holes determinant}]\label{lem:decouple}\mbox{}\\
We concatenate the $n$ lists of points $\left( \bw_{j,k}\right)_{k=1 \ldots M_j}$ (with $j=1\ldots n$) into a single list $\bW = (\bw_1,\ldots,\bw_M)$ of cardinal $M$ to define the $M\times M$ matrix 
\begin{equation}\label{eq:big mat}
\cK^M := \left( \frac{\pi}{N} K_{\infty} (w_j,w_k)\right)_{1\leq j, k \leq M}.
\end{equation}
Under Assumptions~\ref{asum:onehole} and~\ref{asum:multhole} we have that, for an exponentially small $o_N(1)$, 
\begin{equation}\label{eq:decouple fin}
\log \det_{M\times M} \cK ^M = \sum_{j=1} ^n \log \det_{M_j\times M_j} \left(\frac{\pi}{N} K_{\infty} (\bw_{j,k},\bw_{j,\ell})_{1\leq k,\ell \leq M_j}  \right) +o_N (1)
\end{equation}
where, by contrast with~\eqref{eq:big mat}, we use the labeling of points into several different groups.
\end{lemma}

\begin{proof}
 We define
 \begin{equation}\label{eq:big mat bis}
\cK ^{M,j} := \left( \frac{\pi}{N} K_{\infty} (w_\ell,w_k)\right)_{1\leq \ell, k \leq \sum_{k=1} ^j M_k}
\end{equation}
similarly to $\cK^M$, but concatenating only the first $j$ groups of points. That way in particular $\cK^M = \cK^{M,n}$. It suffices to prove that (all remainders $o_N (1)$ appearing in this proof will be exponentially small)
\begin{equation}\label{eq:decouple}
\log \det \cK^{M,j} = \log \det_{M_j\times M_j} \left(\frac{\pi}{N} K_{\infty} (\bw_{j,k},\bw_{j,\ell})_{1\leq k,\ell \leq M_j}  \right) + \log \det \cK ^{M,j-1} +o_N (1) 
\end{equation}
for all $j=2,\ldots,n$ and iterate this relation. We next fix $j\geq 2$ and prove~\eqref{eq:decouple}. Proceeding by induction we are free to assume 
\begin{equation}\label{eq:decouple int}
\log \det \cK ^{M,j-1} = \sum_{k=1} ^{j-1} \log \det_{M_k\times M_k} \left(\frac{\pi}{N} K_{\infty} (\bw_{k,\ell},\bw_{k,m})_{1\leq \ell,m \leq M_k}  \right) +o_N (1)
\end{equation}
We split the points entering in the definition of $\cK ^{M,j}$ into two groups: the $A$ group consisting of the points 
$$\bw^A_{1}, \ldots, \bw^A_{M_j} = \bw_{j,1},\ldots,\bw_{j,M_j}$$ 
and the $B$ group consisting of the other points, 
$$\bw^B_{1}, \ldots, \bw^B_{N_{j-1}} = \bw_{k,\ell}, \quad k= 1 \ldots j-1, \quad \ell = 1, \ldots, M_k$$
with 
$$N_j = \sum_{k=1} ^j M_k.$$
We then expand the determinant 
$$ \det \cK^{M,j} = \sum_{\sigma \in \Sigma_{N_j}} \sgn (\sigma) \prod_{k=1} ^{N_j} \cK ^{M,j}_{k,\sigma(k)}$$
where the sum is over the permutation group of $N_j$ elements.  For clarity of notation we assume that $M_j \leq N_{j-1}$, with simple modifications to the sequel in case the relation is reversed.

Next we split the previous sum according to the number $m$ of $A$ elements that the permutation $\sigma$ sends to $B$ elements. We will denote 
$$ I_m = (i_1,\ldots,i_m), \quad J_m = (j_1,\ldots,j_m)$$
generic $m$-elements subsets of $\{1,\ldots,M_j\}$ and $\{1,\ldots,N_{j-1}\}$ respectively, and use them to label these inter-groups permutations. Then
\begin{multline*}
\det \cK^{M,j} = \sum_{m=0} ^{M_j} (-1)^m \sum_{I_m} \sum_{J_m} \sum_{\sigma \in \Sigma_{M_j-m}} \sum_{\sigma' \in \Sigma_{N_{j-1}-m}} \sgn(\sigma) \sgn(\sigma') \\ \prod_{k=1} ^m  \Kt \left(\bw_{i_k}^A,\bw_{j_k}^B\right) \prod_{h\in I_m^c} \prod_{f\in J_m^c} \Kt \left(\bw_h^A,\bw_{\sigma(h)} ^A\right) \Kt \left(\bw_f^B,\bw_{\sigma'(f)} ^B\right)  
\end{multline*}
with 
$$ \Kt := \frac{\pi}{N} K_\infty$$
and where 
%the first sums are over $m$ elements subsets $I_m\subset\left\{1,\ldots,M_j\right\}$ and $J_m \subset\left\{1,\ldots,N_{j-1}\right\}$, and
the sums over permutations $\sigma,\sigma'$ are (with an abuse of notation) over the indices of 
$$I_m^c := \left\{1,\ldots,M_j\right\}\setminus I_m$$ 
and 
$$J_m^c :=\left\{1,\ldots,N_{j-1}\right\}\setminus J_m$$ respectively. Grouping some terms we reduce the above to 
\begin{multline}\label{eq:combinatoire}
\det \cK^{M,j} = \sum_{m=0} ^{M_j} (-1)^m \sum_{I_m} \sum_{J_m} \prod_{k=1} ^m  \Kt \left(\bw_{i_k}^A,\bw_{j_k}^B\right) \\ \det_{(M_j - m)\times (M_j - m)} \left( \Kt (\bw_h^A,\bw_{h'} ^A)\right)_{h,h'\in I_m ^c} \, \times \, \det_{(N_{j-1} - m)\times (N_{j-1} - m)} \left( \Kt (\bw_f^B,\bw_{f'} ^B)\right)_{f,f'\in J_m ^c}\\
=: \det_{M_j\times M_j} \left(\frac{\pi}{N} K_\infty (\bw_{j,k},\bw_{j,\ell})_{1\leq k,\ell \leq M_j}  \right) \times \det \cK _{M,j-1} +R_{m\geq 1}
\end{multline}
where we have isolated the $m=0$ term in the last equality. Taking the $\log$ and using 
$$ \log (x+y) = \log (x) + \log \left(1+\frac{y}{x}\right)$$
yields the desired terms from the right-hand side of~\eqref{eq:decouple}, with an error suitably small if we prove that 
\begin{equation}\label{eq:decouple10} R_{m\geq 1} \ll \det_{M_j\times M_j} \left(\frac{\pi}{N} K_{\infty} (\bw_{j,k},\bw_{j,\ell})_{1\leq k,\ell \leq M_j}  \right) \times \det \cK _{M,j-1}
\end{equation}
for large $N$, where $R_{m\geq 1}$ is the sum from the first line of~\eqref{eq:combinatoire}, minus the $m=0$ term. 

Under our assumptions, Lemma~\ref{lem:low det} applies to the two determinants above and gives the lower bound
$$ \det_{M_j\times M_j} \left(\frac{\pi}{N} K_{\infty} (\bw_{j,k},\bw_{j,\ell})_{1\leq k,\ell \leq M_j}  \right) \times \det \cK ^{M,j-1} \geq e^{-C (c - c\log c)N}.
$$
Hence, for sufficiently small $c$, it suffices to prove that 
\begin{equation}\label{eq:decouple11}
 \left|R_{m\geq 1}\right| \leq e^{-CN} 
\end{equation}
for a fixed constant $C>0$. This will imply~\eqref{eq:decouple10}, and inserting in~\eqref{eq:combinatoire} will conclude the proof.

We now turn to the proof of~\eqref{eq:decouple11}. Recall that the points from groups $A$ and $B$ are by definition separated by a minimal, finite distance. As per~\eqref{eq:Kmodule} and Assumption~\ref{asum:multhole} we find that, for any set of indices $I_m,J_m$,
$$\prod_{k=1} ^m  \Kt \left(\bw_{i_k}^A,\bw_{j_k}^B\right) \leq e^{-C d^2 m N}$$
where $d$ is the minimal distance between points of the $A$ and $B$ groups. On the other hand, arguing as in the proof of Lemma~\ref{lem:replace ker}, Hadamard's inequality gives, with an argument similar to~\eqref{eq:use Hadamard},   
$$  \left|\det_{(M_j - m)\times (M_j - m)} \left( \Kt (\bw_h^A,\bw_{h'} ^A)\right)_{h,h'\in I_m ^c} \right| \leq C^{M_j - m} \leq C^{cN}$$
and 
$$ 
\left|\det_{(N_{j-1} - m)\times (N_{j-1} - m)} \left( \Kt (\bw_f^B,\bw_{h'} ^B)\right)_{f,f'\in J_m ^c}\right| \leq C^{N_{j-1} - m}\leq C^{cN}
$$
for all such terms appearing in~\eqref{eq:combinatoire}. We have used that by definition $M_j,N_{j-1}\leq M = cN$. Inserting these bounds in~\eqref{eq:combinatoire} and counting terms with $m$ links from group $A$ to group $B$ we find 
\begin{align*}
 |R_{m\geq 1}| &\leq \sum_{m = 1} ^{M_j} \frac{M_j ! N_{j-1}!}{(M_j-m)!(N_{j-1}-m)!}  C^{2cN} e^{-Cd^2 mN} 
 \\&\leq \sum_{m = 1} ^{M_j} e^{m\log (M_j)} e^{m\log N_{j-1} } e^{-C mN}
 \\&\leq \sum_{m = 1} ^{M_j} e ^{2 m\log M} e^{-C mN} \leq e^{-CN}
\end{align*}
if the constant $r_1$ in Assumption~\ref{asum:multhole} is small enough. Indeed, this assumption implies $c\leq r_1 d^2$. This concludes the proof.
\end{proof}

\subsection{Final calculation}

Lemma~\ref{lem:decouple} will allow to compare $F_N (\ba_1,\ldots,\ba_n)$ to $\sum_{j=1}^n F_N (\ba_j)$, as defined in~\eqref{eq:def free}. Subtracting the appropriate mean-field energies and using results from Section~\ref{sec:MF} will then conclude the proof of Theorem~\ref{thm:multhole}. Let us first give the direct comparison between free energies. We denote 
\begin{equation}\label{eq:int ener}
 \Int^{jk} := \sum_{\ell = 1} ^{M_j} \sum_{m=1} ^{M_k} -\log |\bw_{j,\ell} - \bw_{k,m}| 
 \end{equation}
the Coulomb interaction energy between clusters $j$ and $k$.
 
\begin{proposition}[\textbf{Comparison of multiple-holes and single-holes free energies}]\label{pro:freener}\mbox{}\\
Under Assumptions~\ref{asum:onehole} and~\ref{asum:multhole} we have, with an exponentially small remainder,
\begin{multline}\label{eq:freeeners}
F_N (\ba_1,\ldots,\ba_n) =  \sum_{j=1} ^n F_N (\ba_j) -  \sum_{1\leq j < k \leq n} \Int ^{jk} \\
+ \frac{3N^2}{8} \left( (1-n)  + c^2 - \sum_{j=1} ^n c_j^2  - \frac{2}{3}(1+c)^2 \log (1+c) + \frac{2}{3}\sum_{j=1} ^n (1+c_j)^2 \log (1+c_j)\right)\\
+\frac{n-1}{4} N \log N + \frac{(n-1)}{2}\left(\frac{\log 2\pi}{2} - 1\right)N + \frac{5(n-1)}{24}\log N\\
+\frac{(n-1)}{2}\left(\zeta' (-1) + \frac{\log 2\pi}{2}\right) + \frac{1}{24}\left(\log (1+c) - \sum_{j=1} ^n \log (1+c_j) \right) + o_N (1).
\end{multline}
\end{proposition}

\begin{proof}
Reproducing the proof of Lemma~\ref{lem:exp} to compute $\cZ_N (\ba_1,\ldots,\ba_n)$ we obtain 
\begin{multline*}
2 F_N (\ba_1,\ldots,\ba_n) = -\log \cZ_N (\ba_1,\ldots,\ba_n) \\ 
= -\log \ZGin_{N+M} - \sum_{j=1}^M |\bw_j|^2 + 2 \sum_{1\leq j<k \leq M} \log |\bw_j - \bw_k| + A(M,N) + \log \det_{M\times M} \left( \frac{\pi}{N} K_{N+M} (w_j,w_k) \right)  
\end{multline*}
where we have for now concatenated all points in a single list, as in proofs of the preceding subsection, and $ A(M,N)$ is as in~\eqref{eq:AMN} 
% \begin{align}\label{eq:AMN}
% A(M,N) &= \log \frac{(M+N)!}{N!} - M \log \frac{N}{\pi} \nonumber
% \\&= \frac{1}{2}\log \frac{N+M}{N}+(N+M) \log\frac{N+M}{N} + M(\log \pi - 1) +o_N (1).
% \end{align}
Reorganizing terms and arguing as in the proof of Lemma~\ref{lem:replace ker}  we find 
\begin{multline*}
 2 F_N (\ba_1,\ldots,\ba_n) = -\log \ZGin_{N+M} + A(M,N) - 2 \sum_{j=1}^n \cH_N \left(\bw_{j,1}, \ldots,\bw_{j,M_j}\right) \\
- 2 \sum_{1\leq j < k \leq n} \Int^{jk} + \log \det_{M\times M} \left( \frac{\pi}{N} K_{\infty} (w_j,w_k) \right) +o_N (1) 
\end{multline*}
using the notation~\eqref{eq:ref energy} and~\eqref{eq:int ener}. Next, using Lemma~\ref{lem:decouple} we have 
$$ \log \det_{M\times M} \left( \frac{\pi}{N} K_{\infty} (w_j,w_k) \right) = \sum_{j=1} ^n \log \det_{M_j\times M_j} \left(\frac{\pi}{N} K_{\infty} (\bw_{j,k},\bw_{j,\ell})_{1\leq k,\ell \leq M_j}  \right) +o_N (1).$$
Using Lemma~\ref{lem:replace ker} once more thus leads to 
\begin{multline*}
 2 F_N (\ba_1,\ldots,\ba_n) = -\log \ZGin_{N+M} + A(M,N) - 2 \sum_{j=1}^n \cH_N \left(\bw_{j,1}, \ldots,\bw_{j,M_j}\right) \\ - 2 \sum_{1\leq j < k \leq n} \Int ^{jk} + \sum_{j=1} ^M \log \det_{M_j\times M_j} \left(\frac{\pi}{N} K_{\infty} (\bw_{j,k},\bw_{j,\ell})_{1\leq k,\ell \leq M_j}  \right) +o_N (1). 
\end{multline*}
We next use Lemma~\ref{lem:exp} ``backwards'' to deduce 
\begin{multline}\label{eq:pre final}
  2 F_N (\ba_1,\ldots,\ba_n) = 2 \sum_{j=1} ^n F_N (\ba_j) - 2 \sum_{1\leq j < k \leq n} \Int ^{jk} -\log \ZGin_{N+M} + A(M,N) \\ + \sum_{j=1} ^n \left( \log \ZGin_{N+M_j} - A (M_j,N)\right) +o_N (1).
\end{multline}
Combining~\eqref{eq:AMN} and~\eqref{eq:Ginibre part 3} we obtain, for any $J\propto N$,  
\begin{multline*}
-\log \ZGin_{N+J} (0;0) + A(J,N) = \frac{3}{4} (N+J) ^2 - \frac{(N+J)^2}{2} \log \frac{N+J}{N} \\ - \frac{1}{2}(N+J) \log N - \left(\frac{\log 2\pi}{2} - 1\right)(N+J) \\
- \frac{5}{12} \log N - \frac{5}{12} \log \frac{N+J}{N} - \zeta'(-1) - \frac{\log 2\pi}{2} + o_N(1).
\end{multline*}
Using the above for $J=M$ and $J=M_j$, $j=1 \ldots n$, recalling that $M=cN, M_j = c_j N$ with $\sum_j M_j = M$ leads to 
\begin{multline*}
-\log \ZGin_{N+M} + A(M,N)  + \sum_{j=1} ^n \left( \log \ZGin_{N+M_j} - A (M_j,N)\right) = \\
\frac{3N^2}{4} \left( (1-n)  + c^2 - \sum_{j=1} ^n c_j^2  - \frac{2}{3}(1+c)^2 \log (1+c) + \frac{2}{3}\sum_{j=1} ^n (1+c_j)^2 \log (1+c_j)\right)\\
+\frac{n-1}{2} N \log N + (n-1)\left(\frac{\log 2\pi}{2} - 1\right)N + \frac{5(n-1)}{12}\log N\\
+(n-1)\left(\zeta' (-1) + \frac{\log 2\pi}{2}\right) + \frac{1}{12}\left(\log (1+c) - \sum_{j=1} ^n \log (1+c_j) \right) + o_N (1).
\end{multline*}
Inserting in~\eqref{eq:pre final} we finally obtain~\eqref{eq:freeeners}.
\end{proof}

There remains to combine with the mean-field considerations of Section~\ref{sec:MF} to conclude the

\begin{proof}[Proof of Theorem~\ref{thm:multhole}]
 Let $\Eel (\ba_1,\ldots,\ba_n),\Eel (\ba_j)$ be the mean-field energies with all clusters of pinned charged present (respectively, with only the $j$-th one present), as defined in Section~\ref{sec:MF} with $J=N$. Subtracting $\Eel (\ba_1,\ldots,\ba_n)$ from both sides of~\eqref{eq:freeeners}, adding and subtracting $\sum_{j=1}^n  \Eel (\ba_j)$ to the right-hand side there only remains to observe that 
 \begin{multline}\label{eq:split MFE 2}
\Eel (\ba_1,\ldots,\ba_n) - \sum_{j=1}^n  \Eel (\ba_j)  =  - \sum_{1\leq j < k \leq n} \Int ^{jk} \\
+\frac{3N^2}{8} \left( (1-n)  + c^2 - \sum_{j=1} ^n c_j^2  - \frac{2}{3}(1+c)^2 \log (1+c) + \frac{2}{3}\sum_{j=1} ^n (1+c_j)^2 \log (1+c_j)\right).
 \end{multline}
This follows from inspection of~\eqref{eq:split MFE} with an induction on $n$. Each induction step is identical to the $n=2$ one, modulo changing notation. Consider then two clusters of $M_1 = c_1 N$  and $M_2 = c_2 N$ points, corresponding radii in~\eqref{eq:screening subset} 
$$ R_1 ^2 = 1+c_1, \quad R_2 ^2 = 1+c_2, \quad R_{12} ^2 = 1+ c = 1+c_1 +c_2 $$
and constants~\eqref{eq:MFsol2}. Comparing~\eqref{eq:split MFE} with~\eqref{eq:split MFE 2} we need to show that 
 \begin{align}
-\frac{3}{8}   &+ \frac{3}{8} c^2 - \frac{3}{8} \sum_{j=1} ^2 c_j^2  - \frac{1}{4}(1+c)^2 \log (1+c) + \frac{1}{4}\sum_{j=1} ^2 (1+c_j)^2 \log (1+c_j) \nonumber\\
&= \frac{1}{2N^2}\left(C_{R_{12}} -C_{R_{1}} - C_{R_{2}}\right) \label{eq:0}\\
&-\frac{1}{2}\left(\frac{R_{12}^4}{4} -R_{12}^4 \log R_{12} - \frac{R_{1}^4}{4} + R_{1}^4 \log R_{1}- \frac{R_{2}^4}{4}+ R_{2}^4 \log R_{2}\right)\label{eq:1}\\
 &-  c_1 \left(R_{12}^2 \log R_{12} - \frac{R_{12}^2}{2} - R_{1}^2 \log R_{1} + \frac{R_{1}^2}{2}\right) \label{eq:2}\\
 &- c_2 \left(R_{12}^2 \log R_{12} - \frac{R_{12}^2}{2} - R_{2}^2 \log R_{2} + \frac{R_{2}^2}{2}\right)\label{eq:3}.
\end{align}
But, using~\eqref{eq:MFsol2}, the terms on the second line~\eqref{eq:0} give altogether 
$$\frac{1}{2} \left((1+c_1) \log (1+c_1) + (1+c_2) \log (1+c_2) - (1+c) \log (1+c) - 1\right)$$
while those on the third line~\eqref{eq:1} amount to 
\begin{multline*}
\frac{1}{4} (1+c)^2 \log (1+c) - \frac{1}{4} (1+c_1)^2 \log (1+c_1) -\frac{1}{4} (1+c_2)^2 \log (1+c_2) - \frac{1}{8} \left((1+c^2) - (1+c_1^2) - (1+c_2^2)\right)\\
=\frac{1}{4} (1+c)^2 \log (1+c) - \frac{1}{4} (1+c_1)^2 \log (1+c_1) -\frac{1}{4} (1+c_2)^2 \log (1+c_2) + \frac{1}{8} \left(1 - c^2 + c_1^2 + c_2^2\right)
\end{multline*}
and those on the fourth~\eqref{eq:2} and fifth~\eqref{eq:3} lines add up to 
$$
\frac{c}{2} (1+c) - \frac{c}{2} (1+c) \log (1+c) - \frac{c_1}{2} (1+c_1) + \frac{c_1}{2} (1+c_1) \log (1+c_1) - \frac{c_2}{2} (1+c_2) + \frac{c_2}{2} (1+c_2) \log (1+c_2),
$$
leading to the desired identity.
\end{proof}

\newpage

\appendix

\section{Some formulae for the Ginibre ensemble}\label{app:Ginibre}

\subsection{Partition function}

Recall the definition of the Ginibre partition function for $J$ particles in a background of charge density $-4N$ 
\begin{equation}\label{eq:Ginibre def}
\ZGin_J := \int_{\R^{2J}} \prod_{1\leq j < k \leq J} |z_j - z_k| ^2 e ^{-N\sum_{j=1}^J |z_j|^2} dz_1 \ldots dz_J.  
\end{equation}
In the main text we have been using the expansion
\begin{multline}\label{eq:Ginibre part 3}
-\log \ZGin_J = - \log \cZ_J (0,0) \\
= - \frac{J(J+1)}{2} \log \frac{J}{N}
+ \frac{3}{4} J^2 -\frac{1}{2} J \log J - \left( \frac{\log 2\pi}{2} - 1 \right) J \\ 
- \frac{5}{12} \log J -  \zeta'(-1) - \frac{\log (2\pi)}{2} 
+ o_J (1) 
\end{multline}
with an exponentially small remainder.

We recalled in~\cite[Appendix~A]{LamLunRou-22} the well-known formula 
\begin{equation}\label{eq:Ginibre part 0}
\ZGin_J = \frac{\pi^J \prod_{k=1}^J k!}{N^{J(J+1)/2}}
\end{equation}
For $N=J$, we have (cf e.g.~\cite[Equation~(3.7)]{ByuSeoYan-24})
\begin{multline}\label{eq:Ginibre part 2}
-\log \ZGin_N = - \log \ZN (0,0) = \frac{3}{4} N^2 -\frac{1}{2} N \log N - \left( \frac{\log 2\pi}{2} - 1 \right) N \\ 
- \frac{5}{12} \log N -  \zeta'(-1) - \frac{\log (2\pi)}{2} + o_N (1)
\end{multline}
with $\left|o_N (1)\right| \leq e^{-CN}$. Since 
$$ \ZGin_J = \frac{\pi^J \prod_{k=1}^J k!}{J^{J(J+1)/2}} \left(\frac{J}{N}\right) ^{J(J+1)/2}$$
we deduce that, in the general case $J\neq N$ of a mismatch between particle number and background charge density~\eqref{eq:Ginibre part 3} holds.

\subsection{Correlation kernel}

We collect some bounds on the Ginibre correlation kernel(s) that can be found, inter alia, in~\cite[Section~3]{LamLunRou-22}. First we have~\cite[Equation~(3.2)]{LamLunRou-22}
\begin{equation} \label{eq:Kmodule}
| K_\infty(z,w) | = \frac{N}{\pi} e^{-N|z-w|^2/2}.
\end{equation} 
Also, from~\cite[Equation~(3.6)]{LamLunRou-22}, for all $M\geq 0$
\begin{equation}\label{eq:Kmodule2}
|K_{N+M}(z,w)| \leq \frac{N}{\pi} e^{-N( |z|-|w|)^2 /2}  . 
\end{equation}
If $|z|,|w|\leq 1 - \delta$, starting from~\cite[Equation~(3.14)]{LamLunRou-22} we get
\begin{equation}\label{eq:Kdiff pre}
\left|K_{N+M}(z,w)- K_{\infty}(z,w)\right| \leq C N ^{1/2} e^{-\frac{N}{2} \left(1-|z|+1-|w|\right)}  . 
\end{equation}
because the function $\varphi(x)$ used therein is decreasing and convex, so that $\varphi (x) \geq \varphi (1) + \varphi'(1) (x-1)$. It follows that, for  $|z|,|w|\leq 1 - \delta$
\begin{equation}\label{eq:Kdiff pre 2}
\left|K_{N+M}(z,w)- K_{\infty}(z,w)\right| \leq C N ^{1/2} e^{-CN \left(||z|-1| + ||w|-1|\right)}  . 
\end{equation}
Note that, in the proof of~\cite[Lemma~3.3]{LamLunRou-22}, $n$ was assumed fixed in the limit $N\to \infty$ so that the radius of the droplet for $N+n$ Ginibre particles in a background density $-4N$ was $\sim 1$. For $M\propto N$, adapting the estimates therein we find that 
\begin{equation}\label{eq:Kdiff}
\left|K_{N+M}(z,w)- K_{\infty}(z,w)\right| \leq C N ^{1/2} e^{-CN \left(\left||z|-\sqrt{1+c}\right| + \left||w|-\sqrt{1+c}\right|\right)}.
\end{equation}
if $|z|,|w|\leq \sqrt{1+c} - \delta$. Indeed $\sqrt{1+c}$ is the radius of a Ginibre droplet for $N+M$ particles, $M= cN$.

%\section{Free-energy up to order $1$ for a single hole}\label{app:order1}

\newpage 
% %  
% \bibliographystyle{acm}
% \bibliography{/home/rougerie/Travail/Documentation/Bibtex/biblio-NR_Nov22.bib}

\end{document}